\documentclass[11pt]{article}
\pdfoutput=1
\usepackage[latin1]{inputenc}

\usepackage{amstext,amsmath,amssymb,amsfonts,bbm}
\usepackage{hyperref}
\usepackage{authblk}
\usepackage{caption} 
\usepackage{epsfig} 
\usepackage{subfigure}
\usepackage{color} 
\usepackage{graphicx} 
\usepackage{braket} 
\usepackage{slashed} 
\usepackage{dsfont} 
\usepackage{amsthm}  
\usepackage{psfrag}  
\usepackage{tikz}  

\allowdisplaybreaks[4]

\usetikzlibrary{arrows,automata,positioning}

\usepackage{float}  

\usepackage[lmargin=50pt,rmargin=60pt,tmargin=60pt,bmargin=65pt]{geometry}

\newcommand{\be}{\begin{equation}}
\newcommand{\ee}{\end{equation}} 
\newcommand{\Tr}{{\rm Tr}}
 
\DeclareMathOperator{\im}{\mathrm{i}}

\newcommand{\mba}{\mathbf{a}}
\newcommand{\mbb}{\mathbf{b}}
\newcommand{\mbc}{\mathbf{c}}
\newcommand{\mbd}{\mathbf{d}}

\newcommand{\mbS}{\pmb{\Sigma}}
\newcommand{\mbG}{\pmb{G}}

\newcommand{\mbC}{\pmb{C}}
\newcommand{\mbK}{\pmb{K}}

\newcommand{\cG}{\mathcal{G}}

\newcommand{\cJ}{\mathcal{J}}

\allowdisplaybreaks[4]

\newtheorem{proposition}{Proposition}

\newtheorem{lemma}{Lemma}
\newtheorem{theorem}{Theorem}

\theoremstyle{remark}

\begin{document}

\title{\bf Line of fixed points in a bosonic tensor model}

\author[1]{Dario Benedetti}
\author[2]{Razvan Gurau}
\author[3]{Sabine Harribey}

\affil[1]{\normalsize\it Laboratoire de Physique Th\'eorique (UMR 8627), CNRS, Univ.Paris-Sud, \authorcr
\it Universit\'e Paris-Saclay, 91405 Orsay, France \authorcr
email: dario.benedetti@th.u-psud.fr  \authorcr \hfill }

\affil[2]{\normalsize\it Centre de Physique Th\'{e}orique (UMR 7644), CNRS, \'{E}cole
Polytechnique, 91128 Palaiseau, France and
Perimeter Institute for Theoretical Physics, 31 Caroline St. N, N2L 2Y5, Waterloo, ON,
Canada \authorcr email: rgurau@cpht.polytechnique.fr \authorcr \hfill}

\affil[3]{\normalsize \it \'{E}cole Normale Sup\'{e}rieure de Lyon, 46 All\'{e}e d'Italie, 69007 Lyon, France,
\authorcr \it
 Centre de Physique Th\'{e}orique (UMR 7644), CNRS, \'{E}cole
Polytechnique, 91128 Palaiseau, France and Laboratoire de Physique Th\'eorique (UMR 8627), CNRS, Univ.Paris-Sud, 
\authorcr \it 
 Universit\'e Paris-Saclay, 91405 Orsay, France \authorcr email: sabine.harribey@ens-lyon.fr \authorcr \hfill}

\date{}

\maketitle

\hrule\bigskip

\begin{abstract}

We consider the $O(N)^3$ tensor model of Klebanov and Tarnopolsky \cite{Klebanov:2016xxf} in $d<4$ with a free covariance modified to fit the infrared conformal scaling. We study the renormalization group flow of the model using a Wilsonian approach valid in any $d$ (notably we do not require $d=4-\epsilon$ with small $\epsilon$). At large $N$, the tetrahedral coupling has a finite flow, hence it becomes a free parameter.
 The remaining flow can be parameterized
by two couplings which do not mix. We show that, at leading order in $1/N$ but non perturbatively in the couplings, the beta functions stop at quadratic order in the pillow and double-trace couplings. We find four fixed points which depend parametrically on the tetrahedral coupling. 
For purely imaginary values of the latter we identify a real and \emph{infrared attractive} fixed point.
We remark that an imaginary tetrahedral coupling is in fact natural from the onset as the tetrahedral invariant does not have any positivity property, and moreover in the large-$N$ limit the beta functions depend on the square of the tetrahedral coupling, thus they remain real, as long as the other couplings stay real.

\end{abstract}

\hrule\bigskip

\tableofcontents

\section{Introduction}

Tensor models exhibit a \emph{melonic} large $N$ limit \cite{critical,RTM,Klebanov:2018fzb}, different from both the vector \cite{Guida:1998bx,Moshe:2003xn} and the matrix (planar) \cite{'tHooft:1973jz,Brezin:1977sv,DiFrancesco:1993nw} large $N$ limits.
Although, as algebraic objects, tensors are more complicated than matrices, the melonic limit is in fact simpler than the planar one, as melonic diagrams are a subset of the planar ones. 

Tensor models have been extensively studied in zero dimensions (where they were originally introduced as models of quantum gravity  \cite{Ambjorn:1990ge,Sasakura:1990fs}, and further studied with similar motivation  \cite{color,critical,review}) and in one dimension (e.g.\
\cite{Witten:2016iux,Gurau:2016lzk,Klebanov:2016xxf,Peng:2016mxj,Krishnan:2016bvg,Krishnan:2017lra,Bulycheva:2017ilt,Choudhury:2017tax,Halmagyi:2017leq,Klebanov:2018nfp,Carrozza:2018psc}, see also \cite{Delporte:2018iyf,Klebanov:2018fzb} for reviews) as they provide an alternative to the Sachdev-Ye-Kitaev model \cite{Sachdev:1992fk, Kitaev2015, Maldacena:2016hyu, Polchinski:2016xgd,Gross:2016kjj} dispensing with the quenched disorder of the latter.

Proper field theories based on tensor models have been less explored, but they have already been shown to give rise at large $N$ to a new family of conformal field theories 
\cite{Giombi:2017dtl,Prakash:2017hwq,Benedetti:2017fmp,Giombi:2018qgp,Benedetti:2018ghn} which are analytically accessible.   
One for instance has an explicit solution for the infrared two-point function and 
a list of the scaling dimensions of the bilinear operators. The first result is derived from the Schwinger-Dyson equation (SDE), while the second from the 
Bethe Salpeter equation (BSE). The treatment of the two equations is remarkably similar \cite{Giombi:2017dtl,Bulycheva:2017ilt,Giombi:2018qgp,Klebanov:2018fzb}:
\begin{itemize}
 \item at large $N$ both equations truncate to the first non trivial term (the fundamental melon for the SDE and 
the one rung ladder kernel for the BSE). 
\item in both cases one neglects in the infrared the free term and solves the equation self consistently.
\end{itemize}

However, the results obtained by this method are somewhat formal, as both the SDE and the BSE have divergences. For fermionic models, some of the divergences (like for instance the mass) are tamed by anticommutation.
However, no such mechanism works for bosonic models. So far these divergences have been treated by dimensional regularization. 

The aim of this paper is to treat melonic conformal field theories rigorously, and in order to deal with the divergences that appear in the perturbative expansion, we use the Wilsonian renormalization group picture.
As we aim to describe  the infrared CFT of \cite{Klebanov:2016xxf,Giombi:2017dtl}, we consider from the onset a free covariance which reproduces the infrared scaling of the two-point function, and which renders the interactions marginal. A similar idea has been applied to the SYK model by Gross and Rosenhaus in \cite{Gross:2017vhb}. One of the main differences of our model to that of Gross and Rosenhaus is that we have not just one marginal interaction but three (as in \cite{Giombi:2017dtl}): while we find that at large $N$ one of them remains exactly marginal, the other two have a non-trivial renormalization group flow, and in order to find a CFT we need to look for fixed points. We prove rigorously the existence of an infrared fixed point of the RG flow

\subsection{Outline of results}
  
 Our results are the following. We consider the $O(N)^3$ tensor model of Klebanov and Tarnopolsky \cite{Klebanov:2016xxf,Giombi:2017dtl}, but with a quadratic part 
 $(-\Delta)^{ \zeta}$ with $\zeta = d/4$, which reproduces the conformal scaling. The model has three couplings: the ``tetrahedral'', ``pillow'', and ``double-trace'' couplings, denoted $\lambda$, $\lambda_p$, and  $\lambda_d$, respectively (see equation \eqref{eq:action}). We show that in the $N\to \infty$ limit but non perturbatively (i.e. at all orders) in the coupling constants the RG flow has four lines of fixed points parameterized by $\lambda$. In detail, we show that for any $\lambda$:
\begin{description}
 \item[\bf Wave function.] For any bare couplings $\lambda_p,\lambda_d$ (and $\lambda$), there exists a choice of the bare mass $m$ such that,  up to terms which vanish when sending the ultraviolet cutoff $\Lambda$ to infinity and the infrared cutoff 
 $k$ to zero, the effective two-point function is:\footnote{In keeping with standard notation, we denote $\Gamma$ both the Euler Gamma function and various one or two particle irreducible effective actions. }
 \be
  G(p)  = 
  \frac{1}{Z p^{2\zeta}} \;, \qquad   Z^4 - Z^3  =   \lambda^2 \frac{1}{(4\pi)^d } \;\frac{\Gamma \left( 1 -\frac{d}{4} \right) }{ \frac{d}{4}\Gamma\left( 3 \frac{d}{4}\right)} \;,
 \ee
that is, the renormalized mass can be tuned to zero and the wave function renormalization is a finite rescaling. This should come as no surprise: we have fixed the scaling of the covariance to the infrared scaling, hence we do not get an additional anomalous scaling from a wave function renormalization.

  \item[Tetrahedral coupling.] The tetrahedral coupling has a finite flow: in the $\Lambda\to \infty, \,k\to 0$ limit the effective coupling is just a rescaling of the bare one by the wave function constant:
  \[ g = Z^{-2} \lambda \;, \qquad 
  \beta_g = k \frac{\partial g}{\partial k}  = 0 \;.\]
  In particular, denoting 
  \be \label{eq:gc}
  g_c^{-2} = \Gamma \left( 1 -\frac{d}{4} \right) \left[ (4\pi)^d  \, \frac{d}{4}\Gamma\left( 3 \frac{d}{4}\right) \right]^{-1} \;,
  \ee
 the wave function and the bare tetrahedral coupling write in terms of the renormalized one as:
  \[
   Z= \frac{1}{1 - \frac{g^2}{g_c^2}} \; \qquad \lambda = gZ^2 \;.
  \]
\begin{figure}[ht]
\begin{center}
\includegraphics[height=3cm]{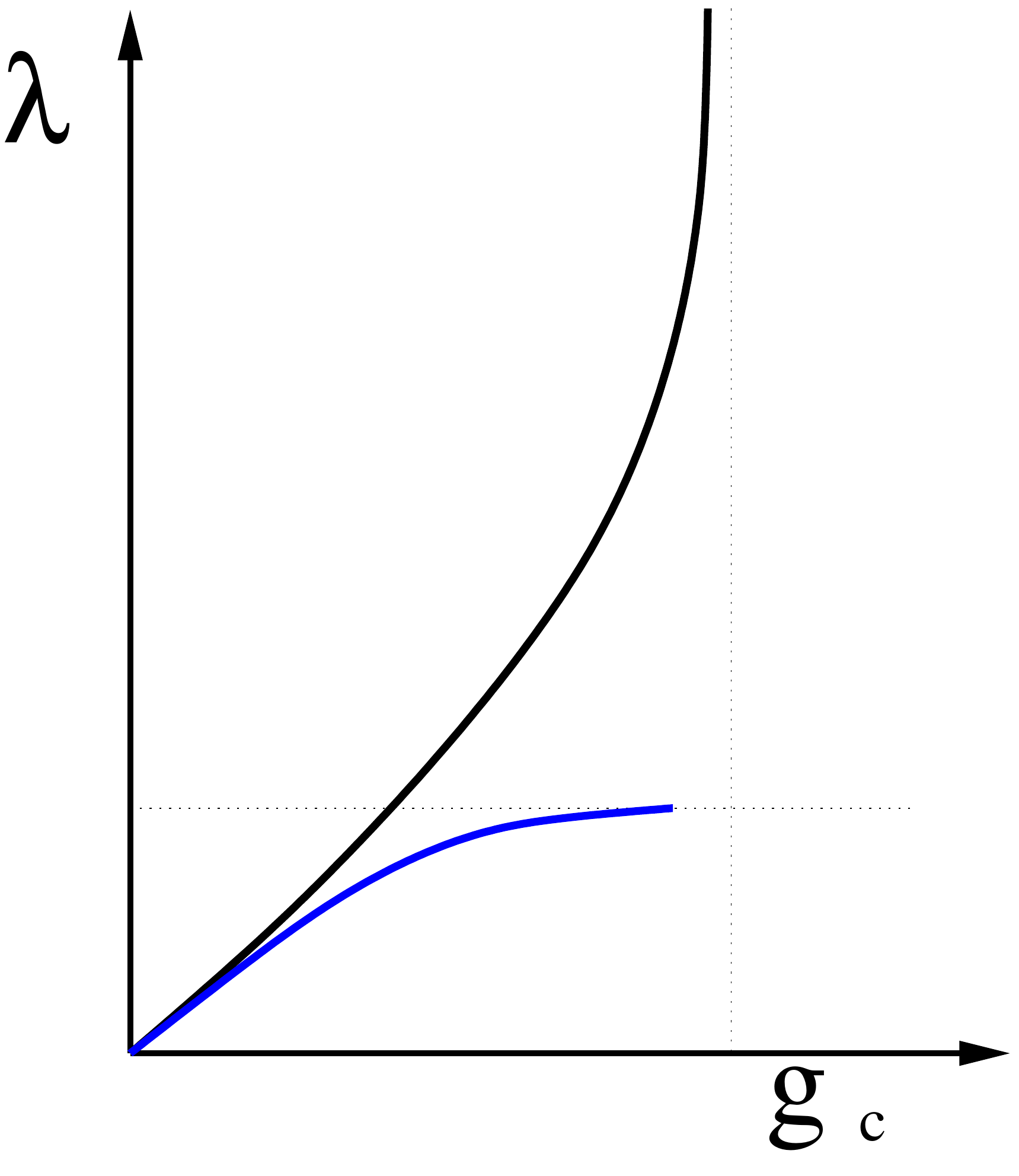} 
 \caption{The bare coupling as a function of the renormalized one. We represented in black the case $\lambda$ real, and
 in blue the absolute value in the case $\lambda$ purely imaginary.} \label{fig:tetra}
 \end{center}
\end{figure}
We distinguish two cases (see Fig.~\ref{fig:tetra}): $\lambda$ real and $\lambda$ purely imaginary:
   \begin{itemize}
    \item \emph{$\lambda$ (and $g$) real.} In this case $\lambda(g)$ is invertible to $g(\lambda)$ for any $\lambda $, $g(\lambda) <g_c$, $g$ asymptotes to $g_c$ and Z diverges when $\lambda\to \infty$ ($g\to g_c$) . 
    \item \emph{$\lambda$ (and $g$) imaginary.} In this case $\lambda(g)$
    is invertible to $g(\lambda)$ for $|\lambda |< 3^{3/2}2^{-4}g_c$, corresponding to  $| g | < 3^{-1/2}g_c$, and $Z$ is bounded.
   \end{itemize}
  
  \item[Pillow and double-trace couplings.] We parameterize the couplings by $\lambda_1 = \lambda_p/3$ and $\lambda_2 = \lambda_p + \lambda_d$. The  $\beta$ functions of the respective renormalized couplings $g_1$ and $g_2$ are independent and \emph{quadratic}. Let us rescale the couplings to $\tilde g = g (4\pi)^{- d/2} \Gamma(\zeta)^{-2} $ and by some abuse of notation drop the tilde. We have:
  \begin{align}
 \beta_{g_1} & =  k\frac{ \partial g_1}{ \partial k}\Big{|}_{\lambda,\lambda_1}     
   =  \beta_0^{g} - 2 \beta_1^{g} \,  g_1 +
   \beta_2^{g} \, g_1^2 \;, 
    \crcr
  \beta_{g_2} & =  k\frac{ \partial g_2}{ \partial k}\Big{|}_{\lambda,\lambda_2}     
   =  \beta_0^{ \sqrt{3} g }- 2 \beta_1^{\sqrt{3} g } \,  g_2 +
   \beta_2^{ \sqrt{3} g }\,  g_2^2 \;,
  \end{align}
where $   \beta_0^{ g}, \beta_1^{ g } $  and   $\beta_2^{g}$ are power series in $g^2$. At first orders they are:
\begin{align*}
 \beta_0^{g} = 
 \left( 2\frac{\Gamma(\frac{d}{4})^2}{ \Gamma(\frac{d}{2})}\right) g^2+  \mathcal{O}(g^4) \;,\qquad
 \beta_1^{ g} =  \mathcal{O}(g^2) \;, \qquad
  \beta_2^{ g} = \left( 2\frac{\Gamma(\frac{d}{4})^2}{ \Gamma(\frac{d}{2})}\right) + \mathcal{O}(g^2) \;. 
\end{align*}

The beta function $\beta_{g_1}$ admits two fixed points:
   \begin{align} \label{eq:FP}
 g_{1\pm} =  \frac{
    \beta_1^{g} 
    \pm \sqrt{ (\beta_1^{g})^2 -\beta_0^{g}\beta_2^{g} }
    }{\beta_2^{g}}  = \pm\sqrt{-g^2} +  \mathcal{O}(g^2) \; ,
   \end{align} 
  and the corresponding critical exponents are:
 \be  \label{eq:exp}
  \beta'_{g_1}( g_{1\pm}) = \pm 2\sqrt{ (\beta_1^{g})^2 -\beta_0^{g}\beta_2^{g} }
     = \pm \sqrt{-g^2} \left( 4\frac{\Gamma(\frac{d}{4})^2}{ \Gamma(\frac{d}{2})}\right) +  \mathcal{O}(g^3)\;.
 \ee
 The beta function $\beta_{g_2}$ admits two fixed points and critical exponents of the same form, with $g\to\sqrt{3}g$.
 Hence, the model has four fixed points in total, each of them actually defining a line parameterized by $g$ in the complex $\{g_1,g_2\}$ space. For $g\to 0$, they all merge into a trivial fixed point: for $g=0$, non-trivial fixed points can only be obtained by moving away from marginality (i.e.\ by taking $4\zeta-d=\epsilon>0$).
 
 \item[Imaginary tetrahedral coupling.]
 
  Contrary to the pillow and double-trace invariants, the tetrahedral invariant does not have any positivity property. Furthermore, due the melonicity of the large-$N$ limit, the beta functions depend on $g^2$. Thus, we can consider
  a purely imaginary tetrahedral coupling $g = \pm\im |g|$, in which case the fixed point values above are real, at least for small $g$. In particular, $g_{1+}>0$ and $\beta'_{g_1}( g_{1+}) >0$, that is, $(g_{1+},g_{2+}) $ is an \emph{infrared attractive} fixed point.

 \item[Dimension of bilinear operators.] For imaginary tetrahedral coupling, we obtain a
 \emph{real} spectrum of bilinear scalar operators. The dimensions of the 
 operators are:
\be
h_{0\pm} = \frac{d}{2} \pm \alpha_0 |g| +  \mathcal{O}(g^3) \;, \qquad 
h_n = \frac{d}{2} + \alpha_n |g|^{2} + \mathcal{O}(g^3) \;, \;\;\; n
\in \mathbb{N}^+ \;,
\ee
with both $\alpha_n$ for $n\ge 1$ \emph{and} $\alpha_0$ real. 

  \end{description}
  
   \bigskip
   
   The fixed points we describe here are very different from the usual Wilson-Fisher fixed point. Let us compare our results with the Wilson-Fisher type of fixed point identified in \cite{Giombi:2017dtl} in the case of the same tensor model but  with $\zeta=1 $ instead of $\zeta=d/4$:
   \begin{itemize}
    \item  the Wilson-Fisher-like fixed point is reliable only for small $\epsilon =4-d$, while our results apply in any $d<4$. Our control parameter is the (bare or renormalized) tetrahedral coupling itself and not $\epsilon$. 
    \item at the Wilson-Fisher-like fixed point one gets an anomalous scaling dimension of the field, while in our case the scaling dimension of the field is fixed (although non-canonical).
    \item the Wilson-Fisher-like fixed point relies on the cancellation of the mass dimension of the coupling with the radiative corrections.\footnote{From \cite{Giombi:2017dtl}, the beta function for the tetrahedral coupling in units of cutoff reads $\beta_g=-\epsilon g+2 g^3$.} This is unlike our case, as we deal with genuinely marginal couplings in any $d$.
    \item because the mechanism of the Wilson-Fisher-like fixed point 
    requires to cancel the mass dimension of the tetrahedral coupling, the fixed point value of the tetrahedral coupling is real for $\epsilon>0$ and consequently the pillow and double-trace ones are purely imaginary. This is the origin of the instability of the fixed point discussed in 
    \cite{Giombi:2017dtl}. An imaginary tetrahedral coupling, and thus real  pillow and double-trace ones, can in their case be obtained for $\epsilon<0$, i.e.\ for $d>4$, but then one deals with an ultraviolet fixed point. Furthermore, the spectrum of scalar bilinear operators computed in \cite{Giombi:2017dtl} shows an upper limit $d=4.155$ beyond which complex dimensions reappear. 
    
    As the tetrahedral invariant has no positivity property, contary to \cite{Giombi:2017dtl},  we have the freedom to consider an imaginary tetrahedral coupling.    
    In this case we find instead a real IR fixed point with real exponents for any 
    $d<4 $ as long as $|g|< g_*$ for some critical coupling $g_*$, as we will discuss in section \ref{sec:spectrum}.

\end{itemize}

\paragraph{Conformal window.} Our results should be compared\footnote{We would like to thank I. Klebanov for pointing out to us reference \cite{Kim:2019upg} and the parallel between our results and theirs.} to the ones of \cite{Kim:2019upg} where the authors consider a one dimensional model with two Majorana fermions and $O(N)^3$ invariance. Their model has no pillow or double trace couplings, but it has several tetrahedral couplings whose relative strength can be dialed up by tuning a parameter. 
In that model the tuning parameter has a critical value where the conformal dimension of an off-diagonal ``mass'' bilinear $\psi_1 \psi_2$ becomes complex,
$d/2+\im \alpha$, with real $\alpha$. Beyond the critical value the ``mass'' bilinear acquires a non zero vacuum expectation value which spontaneously breaks conformal invariance, as well as the discrete symmetries of the model, suggesting a second order phase transition between broken and unbroken symmetry phases. 

Complex dimensions appear also in our case (see section \ref{sec:spectrum}). For $-g_*^2<g^2<0$ our critical exponents are real, but for $g^2>0$ they become of the form $d/2+\im \alpha$ as in  \cite{Kim:2019upg}, while for $g^2< -g_*^2$ they become complex again, but with a real part different from $d/2$. 
Remembering that in the AdS/CFT dictionary \cite{Gubser:1998bc,Witten:1998qj}, $h_{\pm}=d/2 \pm \sqrt{ d^2/ 4 +m^2}$  with $m$ being the mass of a field in AdS${}_{d+1}$, we have the following interpretation. 

The complex dimensions in \cite{Kim:2019upg}, and in our model for $g^2>0$, from the bulk point of view are due to particles which violate the Breitenlohner-Freedman bound $m^2\geq -\frac{d^2}{4}$ \cite{Breitenlohner:1982jf}. 
It is likely that  in our case the mass bilinear also acquires a non zero VEV, but checking this properly is quite involved and we postpone it for further work. This would in particular support the conjecture formulated in Section 3 of \cite{Kim:2019upg}. However, contrary to \cite{Kim:2019upg}, in our case a nontrivial VEV of the mass bilinear does not break any of the symmetries of the model. Thus we expect only a spontaneous breaking of conformal invariance, similar to what happens in the vector $\varphi^6$ model in three dimensions \cite{Bardeen:1983rv,Amit:1984ri}.

On the other hand, for $g^2< -g_*^2$ it seems that the complex dimensions are associated to particles with complex masses in the bulk, ${\rm Im}(m^2)\neq 0$.
However, since $g_*$ is always greater or equal than the maximal value of $|g|$ for which $\lambda(g)$ is invertible to $g(\lambda)$ (see Fig.~\ref{fig:tetra} for the case of  imaginary coupling), the bulk instability in such case is probably related  to the impossibility to define the renormalized model in such range of the tetrahedral coupling.

\paragraph{Plan of the paper.} In section \ref{sec:model} we introduce in detail the model, its expansion in Feynman graphs, and the 2PI formalism, which neatly captures the resummed $n$-point functions at  large $N$. In section \ref{sec:RG} we review the Wilsonian renormalization group formalism that is the backbone of our construction. In sections \ref{sec:2point} and \ref{sec:4point} we construct and renormalize the two- and four-point functions, thus obtaining the beta functions in  section \ref{sec:4point-beta}. In section \ref{sec:divergences} we discuss in detail the coefficients of the beta functions to all orders in $g$.  Lastly, in section \ref{sec:spectrum} we study the spectrum of bilinear operators at the IR fixed point by conformal field theory methods. We close with an appendix detailing some explicit computations.


\section{The bosonic CTKT model}
\label{sec:model}

We will deal in this paper with a modified version of the $O(N)^3$ model of Klebanov and Tarnopolsky \cite{Klebanov:2016xxf}. As the zero dimensional version of the model 
has been introduced by Carrozza and Tanasa \cite{Carrozza:2015adg}, we will henceforth refer to it as the CTKT model.

We consider a real tensor field of rank $3$, $\varphi_{a_1a_2 a_3}(x)$,
transforming under $O(N)^3$ with indices distinguished by the position, and we denote $\mba = (a^1,a^2,a^3)$.
The action of the model is:\footnote{From now on repeated indices are summed.
We work in $d$ space dimensions, we denote $x,y$ and so on positions, $\int_x \equiv \int d^dx$ and
$p,q$ and so on momenta and $\int_p \equiv \int \frac{d^dp}{(2\pi)^d}$. The Fourier transform is 
$  f(p) = \int_x  e^{\im p x} f(x)$ with inverse $f(x) = \int_p  e^{-\im px} f(p)$; we denote them by the same symbol, but context and argument of the function should lift any ambiguity.
The operator product in momentum space is $\int_q f(p,q) h(q,r)$, the identity operator has kernel $ (2\pi)^d \delta(p-q)$,
and translation invariant operators in the direct space are diagonal in momentum:
\begin{align}
H(x,y) & = H(x-y) = \int_p \;  e^{ - \im p (x-y) } H(p) \;, \qquad H(p) = \int_u e^{\im p u} H(y+u,y)\crcr
H(p_1,p_2) & = \int_{x,y} e^{ - \im p_1 x - \im p_2 y} \; H(x,y) = (2\pi)^d \delta(p_1+p_2) H(p_2)
 \; . 
\end{align} 
}
\be\label{eq:action} 
\begin{split}
    S[\varphi]  & =   \frac{1}{2} \int d^dx \;   \varphi_{\mba}(x) (   - \Delta)^{\zeta}\varphi_{\mba}(x) + S^{\rm int}[\varphi]\;,\crcr
   S^{\rm int}[\varphi]  & = \frac{ m^{2\zeta}}{2} \int d^dx \;   \varphi_{\mba}(x) \delta_{\mba \mbb} \varphi_{\mbb}(x)  + 
   \frac{ \lambda }{4 N^{3/2}} \int d^d x \;   \delta^t_{\mba \mbb\mbc\mbd} \; \varphi_{\mba}(x) \varphi_{\mbb}(x)  \varphi_{\mbc}(x) \varphi_{\mbd }(x)\crcr
       & \qquad +   \int d^d x \;  \left( \frac{ \lambda_p }{4 N^{2}} \; \delta^p_{\mba\mbb; \mbc\mbd} +  \frac{ \lambda_d }{4 N^{3}}  \; \delta^d_{\mba\mbb; \mbc\mbd } \right) \; \varphi_{\mba}(x) \varphi_{\mbb}(x)  \varphi_{\mbc}(x) \varphi_{\mbd }(x) 
     \; ,     
\end{split}
\ee 
where $\Delta=  \partial_{\mu}\partial^{\mu}$,  $\delta_{\mba \mbb}  = \prod_{i=1}^3 \delta_{a^i b^i} $ and:
\begin{align}
    \delta^t_{\mba \mbb\mbc\mbd}  = \delta_{a^1 b^1}  \delta_{c^1 d^1} \delta_{a^2 c^2}  \delta_{b^2 d^2 } \delta_{a^3 d^3}   \delta_{b^3 c^3} \;  , \quad
  \delta^p_{\mba\mbb; \mbc\mbd }= \frac{1}{3} \sum_{i=1}^3  \delta_{a^ic^i} \delta_{b^id^i} \prod_{j\neq i}  \delta_{a^jb^j}  \delta_{c^jd^j} \;,
 \quad  \delta^d_{\mba\mbb; \mbc\mbd }  = \delta_{\mba \mbb}  \delta_{\mbc \mbd} \;, 
\end{align}
where $t$ stands for \emph{tetrahedron}, $d$ for \emph{double-trace} and $p$ for \emph{pillow} pattern of contraction. Because it plays a special role below, we have distinguished the coupling $\lambda$ of the tetrahedral
invariant and did not assign any subscript to it.

It is convenient to introduce a graphical representation of the $O(N)^3$ invariants, which also justifies the names of the different contraction patterns. We  represent every tensor as a vertex and every contraction of two indices as an edge. 
We assign to these edges a color $1$, $2$ or $3$ (or red, green, and blue), corresponding to the position of the three indices in the tensor. The quartic invariants of \eqref{eq:action} are represented in Fig.~\ref{fig:interactions}.

\begin{figure}[ht]
\begin{center}
\includegraphics[width=0.7\textwidth]{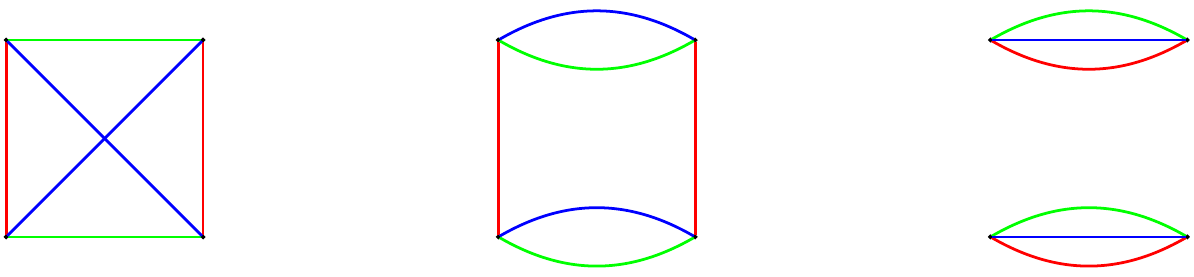}
 \caption{Graphical representation of the quartic $O(N)^3$ invariants. From left to right: the tetrahedron, the pillow, and the double-trace (there are three pillow contractions, distinguished by the color of the vertical edge).} \label{fig:interactions}
 \end{center}
\end{figure}

The CTKT model is obtained for $\zeta=1$, but here we will allow a non trivial power of the Laplacian $\frac{d}{4}\leq \zeta\leq 1$, which preserves  reflection positivity of the propagator and power counting renormalizability of the quartic interactions. And unlike the fermionic CTKT model in one dimension \cite{Klebanov:2016xxf}, where one retains just the tetrahedral interaction, in higher dimensions we have to include all the terms demanded by perturbative renormalizability, hence the mass, pillow, and double-trace terms in \eqref{eq:action}.

To simplify the notation, we sometimes denote $A=(\mba,x)$, $\delta_{AB} = \delta_{(\mba,x) (\mbb,y)} = \delta_{\mba \mbb} \delta(x-y)$ and $\delta(x-y) = \delta_{xy}$.
We denote bilocal operators by bold face. For instance the covariance of the theory $ \mbC$ is:
\be\label{eq:cov} 
\begin{split}
 & \mbC_{AB} = \mbC_{\mba \mbb}(x,y) = \delta_{\mba \mbb}  \; \frac{1}{ (  - \Delta)^{\zeta}  } (x,y) \equiv \delta_{\mba \mbb} \;  C(x,y)  \;, \crcr
 & C(x,y) = C(x-y) = \int_p e^{ - \im p (x-y) } C(p) \;, \qquad C(p) = \frac{1}{p^{2\zeta}}  = \frac{1}{\Gamma(\zeta)} \int_0^{\infty} d\alpha \;\alpha^{\zeta -1 } e^{- \alpha p^2} \;. 
\end{split}
\ee 
The last line can be combined to give the direct space representation:
\be
C(x-y) = \frac{1}{(4\pi)^{d/2} \Gamma(\zeta)} \int_0^{\infty} d\alpha \;\alpha^{\zeta -1-d/2 } e^{- \frac{(x-y)^2}{4 \alpha}}\;,
\ee
which is well defined for $\zeta<d/2$.

We are interested in computing, at leading order in $1/N$  but at all orders in the coupling constants, the connected correlation functions of the theory.
At large $N$ the theory simplifies significantly: the partition function and correlations admit a $1/N$ expansion, as we will now recall. 

\subsection{Feynman graphs}
\label{sec:model-graphs}

The free energy (and the connected $n$-point functions) of the theory can be expanded in connected Feynman graphs $\cG$. 
We will actually use two types of graphs: 4-colored graphs and ordinary Feynman graphs.

The representation as 4-colored graphs is standard in tensor models 
\cite{RTM,uncoloring,Carrozza:2015adg}, and it is obtained as follows. 
Each interaction invariant is represented a 3-colored graph, as above, and we will also call it a \emph{bubble}, as in \cite{RTM,uncoloring,Carrozza:2015adg}. The propagators are represented by edges of a new color connecting the tensors (the vertices of the bubbles), which we call $0$ (black in Fig~\ref{fig:graph}). An example of the resulting 4-colored graphs is given  in Fig.~\ref{fig:graph}.

\begin{figure}[ht]
\begin{center}
\includegraphics[width=0.7\textwidth]{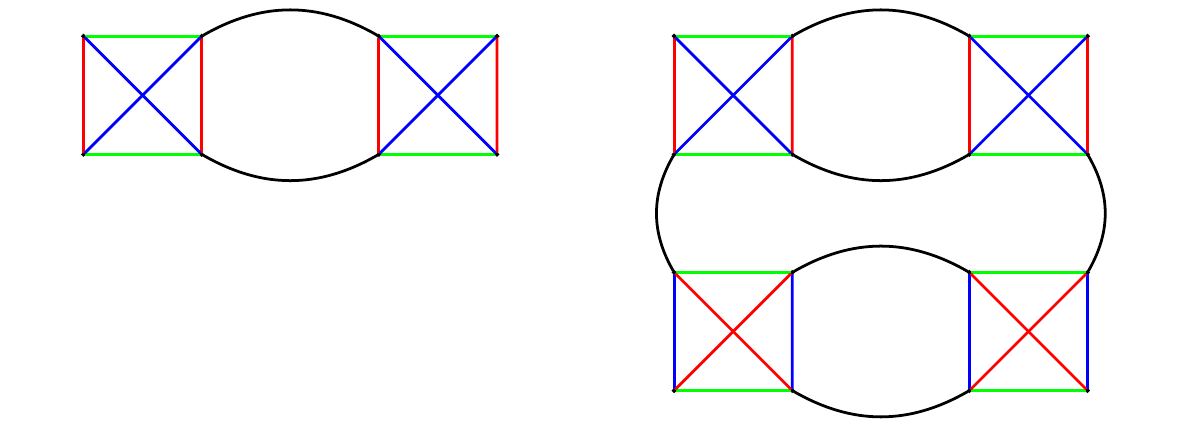}
 \caption{Two Feynman graphs, with external tensor contractions equivalent to the pillow (left) and double-trace (right) invariants.} \label{fig:graph}
 \end{center}
\end{figure}

The ordinary Feynman graphs are obtained by shrinking each bubble to a point (appropriately colored in order to still distinguish the different interaction bubbles, if necessary). An example is given in Fig.~\ref{fig:melontadpoles}, where however we omit the colors of the vertices.

While the ordinary Feynman graphs are simpler, and they are sufficient for representing Feynman integrals (which we will do later), the 4-colored graphs are useful for identifying the correct powers of $N$.
In fact, in a 4-colored graph, each propagator identifies the indices on its two end tensors, hence the indices circulate along the cycles of colors $0i$, which we call faces. We obtain a free sum, that is a factor $N$, per face. 
We denote $n_t(\cG)$, $n_p(\cG)$ and $n_d(\cG)$ the numbers of tetrahedral, pillow, and double-trace bubbles, and $F(\cG)$ the number of faces of $\cG$. We associate a variable $x_v$ to each bubble in $\cG$. The free energy of the model is:
\begin{align}
 {\mathcal F} & =-  \ln \bigg\{  \int [d\varphi]\; e^{-S[\varphi]}\bigg\} \\
 & = \sum_{\cG} N^{F -\frac{3}{2}n_t - 2 n_p - 3n_d}  \frac{\lambda^{n_t}}{ n_t! 4^{n_t}} \frac{\lambda_p^{n_p}}{ n_p! 12^{n_p}}   \frac{\lambda_d^{n_d}}{ n_d! 4^{n_d}} (-1)^{n_t+n_p+n_d+1} A(\cG) \int_x 1  \;, \crcr
 A(\cG)& =  
 \int \prod_{ v\neq v_0 } dx_{v} \prod_{e\in \cG } C(x_e,y_e) \;,
\end{align}
where $\cG$ runs over connected vacuum 4-colored graphs with labelled tensor vertices, $v_0$ is an arbitrary root vertex, and $x_e$ and $y_e$ denote the positions of the end vertices of the edge $e$. 

\paragraph{The $1/N$ expansion.}
The model has a $1/N$ expansion\cite{Carrozza:2015adg,Klebanov:2016xxf}. The simplest way to see this is to observe that pillow and double-trace vertices can be obtained as radiative corrections from the tetrahedral vertex: the pillow is a rung (Fig~\ref{fig:graph}, left), and the double-trace 
is a ladder made out of two rungs with different color inside their loop (Fig~\ref{fig:graph}, right). Replacing the pillow and double-trace vertices in a graph by their minimal resolution in terms of tetrahedral vertices one associates to any graph $\cG$ 
a graph $\hat \cG$ having \emph{only} tetrahedral vertices but the same scaling in $N$:
\[
 F(\cG) -\frac{3}{2}n_t(\cG) - 2 n_{p}(\cG) - 3n_{d}(\cG) =  F(\hat \cG) -\frac{3}{2}n_t( \hat \cG) \;.
\]
Starting from $\hat \cG$ one can build three \emph{jackets} \cite{expansion1,Carrozza:2015adg} ${\cal J}^i$, that is ribbon graphs\footnote{The ribbon graphs are made evident in the stranded representation, where one replaces each black line and vertex by three parallel red, green, and blue lines: a jacket ${\cal J}^i$ is then obtained by simply deleting color $i$.} obtained by ignoring the faces of color $0i$.
Each jacket has a non orientable genus $k({\cal J}^i) \ge 0$ and the number of faces\footnote{It is at this point that one uses the fact that $\hat \cG$ has only tetrahedral vertices. This construction is slightly more complicated on the original graph $\cG$, as the jackets of $\cG$ are not necessarily connected \cite{Carrozza:2015adg}.}
$ {\cal F}(\cJ^i) =  n_t(\hat \cG) + 2 - k({\cal J}^i)$.
As every face belongs to two jackets, the total number of faces of $\hat \cG$ is:
\[
 {\cal F}(\hat \cG) = \frac{3}{2} n_t(\hat \cG) + 3 - \frac{1}{2} \sum_{i} k({\cal J}^i) \; .
\]
Denoting $ \omega(\cG) = \frac{1}{2} \sum_{i} k({\cal J}^i) \ge 0$ the \emph{degree} of the original graph $\cG$, the scaling with $N$ of a connected vacuum graph is:
\[
 N^{3 - \omega(\cG)} \; .
\]

By the standard arguments \cite{critical,RTM} $\cG$ has degree zero if and only if $\hat \cG$ is melonic.
That is the leading order graphs are melonic \emph{after} substituting all the pillows and double-trace vertices by their minimal realizations in terms of 
the tetrahedral vertex. In terms of the original interactions in $\cG$, one gets \emph{melon tadpole} \cite{Benedetti:2017qxl} graphs, that is 
graphs obtained by iterated insertions of melons or tadpoles into melons or tadpoles,
see Fig.~\ref{fig:melontadpoles}. Observe that all the tadpoles are based on
either pillow or double-trace vertices, while the end vertices of the melons are tetrahedral.
\begin{figure}[ht]
\begin{center}
\includegraphics[width=0.3\textwidth]{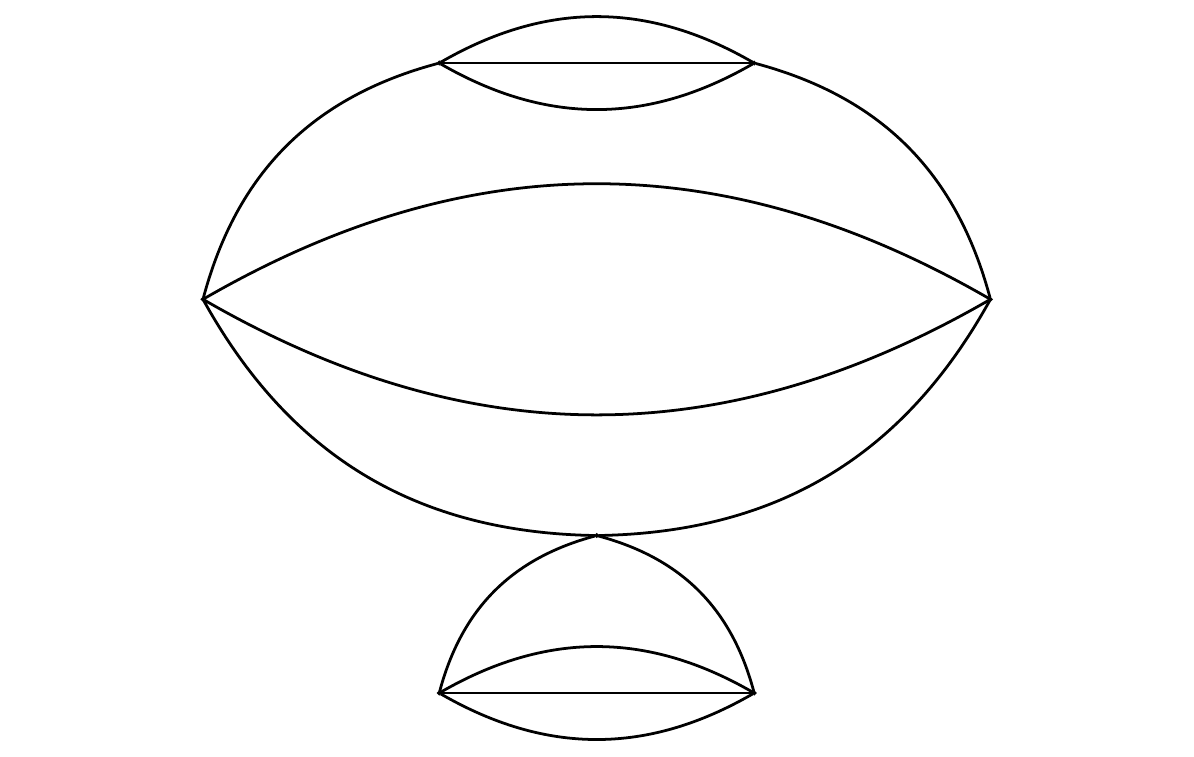}
 \caption{A melon tadpole graph, where all the invariants have been shrunk to point-like vertices.} \label{fig:melontadpoles}
 \end{center}
\end{figure}

\subsection{The 2PI effective action}
\label{sec:model-2PI}

The two-particle irreducible (2PI) effective action formalism is particularly well adapted to the tensor $1/N$ expansion \cite{Benedetti:2018goh}.\footnote{In fact, the $1/N$ expansion offers a controlled way of implementing as a proper expansion scheme the so-called $\Phi$-derivable truncations studied for example in \cite{Blaizot:2003br,Berges:2005hc,Blaizot:2010zx}.}
First of all, observe that $S[- \varphi] = S[\varphi]$, hence the odd point functions are zero in the absence of spontaneous symmetry breaking, which we will assume in the following. 
We define the generating function with bilocal source $\mbK_{AB}= \mbK_{\mba \mbb}(x,y)$:\footnote{We omit the source term linear in the fields in order to keep the presentation concise, as the bilinear term is enough for our purposes. For the general construction see \cite{Benedetti:2018goh}.}
\be \label{eq:W-2PI}
 e^{ W [\mbK]} = \int d\mu_{ \mbC  }( \varphi )\; e^{-S^{\rm int}[\varphi] + \frac{1}{2} \varphi_A \mbK_{AB} \varphi_B}   \;, 
\ee
where $d\mu_{ \mbC  }$ denotes the normalized Gaussian measure with covariance $\mbC$.
The source effectively shifts the inverse covariance: $\mbC^{-1}\to\mbC^{-1}-\mbK$.
Taking into account that the odd point functions are still zero in the presence of the source, 
the derivatives of $W$ write in terms of the connected two and four-point functions with source $\mbK$\footnote{Note that the derivative of a symmetric function $\mbK_{AB} = \mbK_{BA}$ with respect to itself is the projector on symmetric functions
\[ \frac{\delta \mbK_{AB}}{\delta \mbK_{EF}} = {\cal S}_{AB;EF} = \frac{1}{2} (\delta_{AE} \delta_{BF} + \delta_{AF} \delta_{BE}) \;.\]}:
\begin{align}\label{eq:Wderiv}
2 \frac{\delta  W}{ \delta \mbK_{AB}}  = & \braket{\varphi_{A} \varphi_{B} }^c_{\mbK} \;, \crcr
4 \frac{\delta^2  W}{ \delta \mbK_{AB} \delta \mbK_{EF} }  = & \braket{\varphi_{A}  \varphi_{B} \varphi_{E} \varphi_{F} }^c_{\mbK}  + 
  \braket{\varphi_{A}  \varphi_{E} }^c_{\mbK} \braket{  \varphi_{B} \varphi_{F} }^c_{\mbK} +
  \braket{\varphi_{A }   \varphi_{F} }^c_{\mbK} \braket{  \varphi_{B } \varphi_{E} }^c_{\mbK} \;.
\end{align}

Setting $\mbK=0$, one recovers the connected two and four-point functions of the original theory.
Inverting $ 2 \frac{\delta  W}{ \delta \mbK_{AB}}  = \mbG_{AB} $ yields the source $\mbK[\mbG] $ which ensures that the connected two-point function is exactly $\mbG_{AB}$. 
The Legendre transform of $W$ is:
\be \label{eq:GammaG}
  \Gamma [ \mbG]   =    \bigg\{ -   W[\mbK] + \frac{1}{2} \Tr[\mbG \mbK] \bigg\}_{ \mbK = \mbK[\mbG]}  \;,
\ee
where $\Tr$ denotes a trace over both the indices and the positions. The derivatives of $\Gamma$ are:
\begin{align*}
 &  \frac{\delta  \Gamma}{\delta \mbG_{AB}} = \frac{1}{2} \mbK_{AB} 
 \;, \qquad \frac{\delta^2   \Gamma}{\delta \mbG_{AB} \delta \mbG_{EF}}   =   \frac{1}{2} \frac{\delta \mbK_{AB}}{\delta \mbG_{EF}} 
 = \frac{1}{2} \left(  \frac{\delta \mbG  }{\delta K } \right)^{-1}  = M^{-1} \;,
  \crcr
 & M_{ (AB) ; (EF)}   = 4 \frac{\delta^2   W}{ \delta \mbK_{AB}\delta \mbK_{EF}}  
 =  \braket{\varphi_{A}  \varphi_{B} \varphi_{E} \varphi_{F} }^c_{\mbK[\mbG]}  + 
  \mbG_{AE} \mbG_{BF}  + \mbG_{AF} \mbG_{BE} \;. 
\end{align*} 
The field equations $ \frac{\delta \Gamma}{\delta \mbG} =  0$ are equivalent to $\mbK =0$,   and we denote their solution $\bar \mbG$. 
The on-shell two-point function is diagonal in the tensor indices 
$ \bar \mbG_{AB} = \delta_{\mba \mbb} \bar G(x,y)$. 

Let us denote $ - \Gamma^{2PI}[\mbG]$ the sum of non trivial vacuum 2PI graphs (i.e. which don't disconnect by cutting two edges) with vertices defined by $ S[\varphi]$ and with propagators given by $\mbG$. 
The self energy $\mbS$ (the sum of non trivial one-particle irreducible two-point graphs with propagator $\mbC$) can be obtained as:
\be
\mbS_{AB}[\mbG] =  - 2 \frac{\delta \Gamma^{2PI}[\mbG]}{\delta \mbG_{AB}}  \;, 
\ee
where the derivative selects and cuts an edge and the factor $2$ counts the ways to attach it to the external points.
The derivative of the self energy with respect to the two-point function yields the amputated 2PI four-point kernel \cite{Berges:2004yj}. The 2PI irreducible kernel amputated to the right only is:
\be
 {\cal K}_{A'B' ; EF} =  \mbG_{A' A} \mbG_{ B' B} \frac{\delta \mbS_{AB}}{ \delta \mbG_{EF} } \;.
\ee

The full two-point function obeys the Schwinger-Dyson equation $\mbG^{-1} = \mbC^{-1} - \mbK[\mbG] - \mbS[\mbG]$.
Solving for $\mbK$, we get $
 \frac{\delta  \Gamma}{\delta \mbG} = \frac{1}{2} \mbK  = \frac{1}{2} \mbC^{-1} -  \frac{1}{2}  \mbG^{-1} + \frac{\delta \Gamma^{2PI}}{\delta \mbG} \;,
$
and:
\begin{align}
  \Gamma[\mbG] & = \frac{1}{2} \Tr[ \mbC^{-1} \mbG ] -\frac{1}{2} \Tr\ln(\mbG) +  \Gamma^{2PI}[ \mbG ] \;, \\
 \frac{\delta^2   \Gamma}{\delta \mbG_{AB} \delta \mbG_{EF}}  & =   \frac{1}{2} \frac{\delta \mbK_{AB}}{\delta \mbG_{EF}}  = \frac{1}{2} \mbG^{-1}_{AA'} \mbG^{-1}_{BB'} \bigg( {\cal S} - {\cal K} \bigg)_{A'B'; EF} \;,
\end{align}
with ${\cal S}$ the projector on symmetric functions. Now, as the kernel ${\cal K}_{A'B';EF}$ is symmetric in $A'B'$ (and in $EF$), we have ${\cal K} = {\cal S} {\cal K}$,
and using Eq.~\eqref{eq:Wderiv} we get:
\begin{align} \label{eq:4point}
 \braket{\varphi_{A}  \varphi_{B} \varphi_{E} \varphi_{F} }^c_{\mbK[\mbG]}  =2 \left( \frac{{\cal K}}{1 - {\cal K}} {\cal S} \right)_{ AB; E'F' }  G_{E'E}  G_{F'F} \;.
\end{align}

The terms of the 2PI action can be organized in powers of $1/N$. The scaling in $N$ of a term is obtained by substituting for the two-point function its on shell value $\delta_{\mba \mbb} \bar G(x,y)$.
At leading and next-to-leading order in $N$, the combination of the $1/N$ expansion and the 2PI condition leads to a finite number of graphs:
\begin{itemize}
 \item leading order ($N^3$): a graph with a mass two-valent vertex and one edge, a melon with two tetrahedral vertices, a double tadpole with the pillow vertex and one with the double-trace vertex,
\item  next-to-leading order ($N^{5/2}$): three double tadpoles with the tetrahedral vertex (the three possible choices for closing a tadpole are distinguished by the coloring of the tetrahedron).
\end{itemize}
Thus at leading and  next-to-leading order we get \cite{Benedetti:2018goh}:
\begin{equation}
\begin{split}
 - \Gamma^{2PI}[\mbG] = & - \frac{m^{2\zeta}}{2} \Tr[\mbG]  - \frac{\lambda_p}{4N^2} \int_x \mbG_{( \mba,x)(\mbb,x)} \delta^p_{\mba\mbb ; \mbc\mbd} \mbG_{(\mbc,x)(\mbd,x)} - 
  \frac{\lambda_d}{4N^3} \int_x \mbG_{( \mba,x)(\mbb,x)} \delta^d_{\mba\mbb ; \mbc\mbd} \mbG_{(\mbc,x)(\mbd,x)} \crcr
  & + \frac{1}{2} \left( \frac{\lambda}{4 N^{3/2}}\right)^2 4 \int_{x,y} \delta^t_{\mba\mbb \mbc\mbd} \delta^t_{\mba'\mbb' \mbc'\mbd'}
  \mbG_{( \mba, x)(\mba',y)} \mbG_{( \mbb, x)(\mbb',y)}  \mbG_{( \mbc, x)(\mbc',y)} \mbG_{( \mbd, x)(\mbd',y)}  \crcr
  & -  \frac{\lambda}{4N^{3/2}}   \int_x \mbG_{( \mba,x)(\mbb,x)} \mbG_{(\mbc,x)(\mbd,x)} \bigg(  \delta^t_{\mba\mbb   \mbc\mbd} +  \delta^t_{\mba  \mbc \mbb  \mbd}  +  \delta^t_{\mba \mbc \mbd   \mbb}  \bigg)\;,
\end{split}
\end{equation}
where the first two lines are leading order and the last one is next-to-leading order. The self energy is:
\begin{equation}
\begin{split}
  \mbS_{(\mba,x)(\mbb,y)}  = & -m^{2\zeta} \delta_{\mba \mbb}\delta_{xy} - \frac{\lambda_p}{N^2}  \delta_{xy} \delta^p_{\mba\mbb ; \mbc\mbd} \mbG_{(\mbc,x)(\mbd,x)} 
 - \frac{\lambda_d}{N^3}  \delta_{xy} \delta^d_{\mba\mbb ; \mbc\mbd} \mbG_{(\mbc,x)(\mbd,x)} \crcr
    & + \frac{\lambda^2}{N^3} \delta^t_{\mba\mbc_1 \mbc_2\mbc_3} \delta^t_{\mbb \mbd_1 \mbd_2\mbd_3} \mbG_{( \mbc_1, x)(\mbd_1,y)}  \mbG_{( \mbc_2, x)(\mbd_2,y)} \mbG_{( \mbc_3, x)(\mbd_3,y)} \crcr
    & -  \frac{\lambda}{N^{3/2}} \delta_{xy}\bigg(  \delta^t_{\mba\mbb   \mbc\mbd} +  \delta^t_{\mba  \mbc \mbb  \mbd}  +  \delta^t_{\mba \mbc \mbd   \mbb}  \bigg)  \mbG_{(\mbc,x)(\mbd,x)} \;.
\end{split} 
\end{equation}

Finally, the four-point kernel at leading and next-to-leading order is:
\begin{equation} \label{eq:kernel}
\begin{split}
& {\cal K}_{ (\mba',x')(\mbb',y') ; (\mbc,z)(\mbd,t) }  =  \mbG_{(\mba',x') (\mba,x)  }\mbG_{(\mbb',y') (\mbb,y)  } \bigg[
    - \frac{\lambda_p}{N^2}  \delta_{xy} \delta_{xz} \delta_{xt} \delta^p_{\mba\mbb ; \mbc\mbd} - \frac{\lambda_d}{N^3}  \delta_{xy} \delta_{xz} \delta_{xt} \delta^d_{\mba\mbb ; \mbc\mbd}  \crcr
 & \qquad + \frac{\lambda^2}{N^3} \delta^t_{\mba\mbc_1 \mbc_2\mbc_3} \delta^t_{\mbb \mbd_1 \mbd_2\mbd_3}    \sum_{i=1}^3
     \left( \frac{1}{2}\delta_{xz}\delta_{yt} \delta_{\mbc_i \mbc} \delta_{\mbd_i \mbd} + \frac{1}{2}\delta_{xt}\delta_{yz}  \delta_{\mbc_i \mbd} \delta_{\mbd_i \mbc}  \right) 
    \prod_{j\neq i} 
    \mbG_{( \mbc_j, x)(\mbd_j,y)} \crcr
&  \qquad -  \frac{\lambda}{N^{3/2}} \delta_{xy}  \delta_{xz} \delta_{xt} 
    \frac{ \bigg(  \delta^t_{\mba\mbb   \mbc\mbd} +\delta^t_{\mba\mbb   \mbd\mbc} +  \delta^t_{\mba  \mbc \mbb  \mbd} +  \delta^t_{\mba  \mbd \mbb  \mbc} 
    +  \delta^t_{\mba \mbc \mbd   \mbb} + \delta^t_{\mba \mbd \mbc   \mbb}  \bigg) }{2}
 \bigg] \;.
\end{split} 
\end{equation}

Below we will be interested in evaluating the two and four-point functions on shell where $ \bar \mbG_{AB} = \delta_{\mba \mbb} \bar G(x,y)$.
We will drop the bar on $G(x,y)$ in order to simplify the notation.

\section{Renormalization}
\label{sec:RG}
 
\paragraph{Motivation.} We consider $d < 4$. According to \cite{Klebanov:2016xxf,Giombi:2017dtl}, the model 
in Eq \eqref{eq:action} with $\zeta =1$ 
should have a non trivial conformal infrared limit.
We aim to study rigorously this putative conformal infrared limit. 
In the IR, the full two-point function is expected to acquire a 
non trivial scaling behavior $G(p) \sim p^{- d/2 }$. 
Using the full two-point function as propagator, the theory with interaction $S^{\rm int}$ in Eq.~\eqref{eq:action} exhibits \emph{ultraviolet} divergences 
in any $d$: the two-point\footnote{As we use the resummed two-point function, the graphs of the effective theory do not have any two-point subgraphs. However, the full two-point function 
must satisfy the Schwinger-Dyson equation which does exhibit ultraviolet divergences.} and four-point graphs are ultraviolet divergent.
In order to make sense of the infrared theory one has two options:
\begin{itemize}
 \item set from the beginning  $ \zeta = d / 4 $ that is start from a bare covariance
 that reproduces the infrared scaling of the two-point function. In the SYK model in one dimension this has been studied by
 Gross and Rosenhaus \cite{Gross:2017vhb}. In $d=3$ (with no tensor indices) the choice $\zeta = 3/4 + \epsilon$ yields the Brydges-Mitter-Scoppola model 
 \cite{Brydges:2002wq,Abdesselam:2006qg}.
 \item argue that the ultraviolet divergences are just an artifact of using the infrared ansatz for the two-point function:
 for $d < 4$ the free covariance dominates in the ultraviolet, hence the effective two-point function will behave 
 at large momentum as $p^{-2}$. 
\end{itemize}

The second option is very difficult to implement (even non rigorously) in practice. One would have to consider that the infrared scaling $G(p) \sim p^{-d/2}$ is a good approximation up to some momentum 
scale $\Lambda$. 
Neglecting the higher momenta makes all the correlation functions depend on (and in fact diverge with) the non physical dimensionful parameter $\Lambda$. In order to eliminate this dependence\footnote{The same situation arises in quantum electrodynamics. Although we all agree that QED is not a UV complete theory and in the UV one needs to take into account the rest of the standard model, 
it still makes sense in the infrared to study QED with a cutoff $\Lambda$ and renormalize it. This leads at low energy to some reasonably accurate predictions, like the anomalous magnetic moment of the electron.} 
one still needs to subtract these divergences using a bare theory at scale $\Lambda$ with bare covariance
$C(p) \sim p^{ - d/2}$.

\bigskip

We choose the first option, and from now on we assume $\zeta = d/4$, although we will keep $\zeta$ arbitrary in most formulas, for convenience and generality.

\paragraph{Power counting.} 
Let $\cG$ be a connected amputated Feynman graph with $n(\cG)$ vertices, $E(\cG)$ edges and $r(\cG)$ external points.
In momentum space one counts an independent integral $d^dp$ for every loop and a propagator $p^{-2\zeta}$ for every edge. Under a global rescaling by a factor $t$ of all the momenta, the amplitude is rescaled as:
\[
 t^{ d \left[ E(\cG) - n (\cG)+ 1 \right] - 2\zeta E (\cG)}  = t^{ \left[ d - \frac{r(\cG)}{2}\left( d -2\zeta\right) \right] + (d-4\zeta) n(\cG) } = t^{d \left( 1 - \frac{ r(\cG)} {4} \right)  } \;,
\]
where $E(\cG) - n (\cG)+ 1$ is the number of loops, and we used $2E(\cG) = 4n(\cG)-r(\cG)$ and $\zeta = d/4$. 
The theory is marginal, that is, the power counting does not depend on the number of internal vertices. The two-point graphs are power divergent  ($t^{\frac{d}{2}}$) in the UV and the four-point graphs are logarithmically divergent in the UV. Graphs with more than six external points are naively UV convergent.
 
\subsection{Wilsonian renormalization group}

\paragraph{The RG transformation.}
In order to access the infrared limit one needs to study the renormalization of the theory. We briefly review the Wilsonian renormalization group  in our setting.\footnote{Our presentation essentially follows the formulation of \cite{Morris:1993qb}; see also \cite{Blaizot:2010zx,Carrington:2012ea,Carrington:2014lba} for the functional RG of the 2PI effective action.} Although standard, we will review some essentials in order to clarify our logic and to highlight some subtleties that arise in our case. We start from \eqref{eq:W-2PI} with an explicit UV cutoff $\Lambda$:
\be 
\begin{split}
 e^{ W [\mbK]} & = \int d\mu_{ C^{\Lambda}}  (\varphi) \; e^{-S_{\Lambda}[\varphi]+ \frac{1}{2} \varphi_A \mbK_{AB} \varphi_B} \;,\\
C^{\Lambda}(p) & = \frac{1}{p^{2\zeta}} \, \Theta\left(\frac{p^2}{\Lambda^2}\right)  \; , \qquad C^{\Lambda}(x) =  \int_p \; \frac{e^{-\imath p x}}{ p^{2\zeta} } \,\Theta\left(\frac{p^2}{\Lambda^2}\right)  \; ,
\end{split}
\ee
where $d\mu_{ C^{\Lambda}  }$ denotes the normalized Gaussian measure with covariance $C^{\Lambda}$. We denote by convention $S_{\Lambda}[\varphi] \equiv S^{\rm int}[\varphi]$ the bare potential of our model \eqref{eq:action}.
The ultraviolet divergences are regularized by the multiplicative 
cutoff function $\Theta(p^2/\Lambda^2)$.\footnote{$\Theta(u)$ is monotonic and takes values between 0 and 1,  such that $\Theta(u)\simeq 1$ for $u<1$, and $\Theta(u)\simeq 0$ for $u>1$.
Typical choices are the exponential cutoff $\Theta(u)= e^{-u}$, the sharp cutoff $\Theta(u)=\theta(1-u)$ (or a smooth approximation of it), and so on.}
While the specific choice of the cutoff function should not affect the main results, we will choose once and for all to use
a normalized upper incomplete Euler gamma function:
\be
  \Theta\left(\frac{p^2}{\Lambda^2}\right) = \frac{\Gamma \left(\zeta; \frac{p^2}{\Lambda^2} \right)}{\Gamma(\zeta)} = \frac{1}{\Gamma(\zeta)}\int_{\frac{p^2}{\Lambda^2}}^{\infty} d\alpha \; \alpha^{\zeta-1} e^{-\alpha} \;,
\ee
which implements a parametric cutoff for the Schwinger parameter $\alpha$, and which for $\zeta=1$ reduces to the standard exponential cutoff.

Let $k\le \Lambda$ be an infrared scale.
The Wilsonian RG transformation consists in integrating out the modes with momenta between $\Lambda$ and $k$, and then rescaling the momenta by $\Lambda/k$ and the fields by their wave function renormalization in order to re-establish the original free covariance of the leftover modes. 

To be more precise, we introduce the slice cutoff function $\chi^{\Lambda}_k(p) = \Theta\left( p^2/\Lambda^2 \right)- \Theta\left( p^2/k^2 \right) $
and we split the covariance as $C^{\Lambda} = C^{k} + C^{\Lambda}_k$, where $C^{k} $ is the covariance with UV cutoff $k $ and $C^{\Lambda}_k$ is the covariance of the fluctuations (the modes with momenta between $\Lambda$ and $k$):
\be
 C^{\Lambda}_k (p) = \frac{1}{p^{2\zeta}} \,\chi^{\Lambda}_k(p) =  \frac{1}{\Gamma(\zeta)}\int_{\Lambda^{-2}}^{k^{-2}}  d\alpha \; \alpha^{\zeta-1} e^{-\alpha p^2 } \; .
\ee

Associated to the split of the covariance, the Gaussian integral also splits as (see for example \cite{Gurau:2014vwa}):
\be \label{eq:1stRGstep}
\begin{split}
 e^{ W [\mbK]}   & = \int d\mu_{C^{\Lambda}}(\varphi )  \;e^{ - S_{\Lambda}[\varphi] + \frac{1}{2} \varphi_A \mbK_{AB} \varphi_B } \\
 & = 
\int d\mu_{C^{ k } }(\varphi)  \int d\mu_{C^{\Lambda}_k}(\chi) \; 
e^{ - S_{\Lambda} [\varphi + \chi ] + \frac{1}{2} (\varphi_A+\chi_A) \mbK_{AB} (\varphi_B+\chi_B)} \\
&= \int d\mu_{C^{ k }} (\varphi ) \;  e^{   W_k [\mbK;\varphi ]}  \;,
\end{split}
\ee
where by a change of variables we can write
\be \label{eq:Wk}
e^{   W_k [\mbK;\varphi ]} = e^{-\frac{1}{2}\varphi_A (\mbC^{\Lambda}_k)^{-1}_{AB}\varphi_B } \int d\mu_{C^{\Lambda}_k}(\chi) \; 
e^{ - S_{\Lambda} [ \chi ] + \frac{1}{2} \chi_A \mbK_{AB} \chi_B + \varphi_A (\mbC^{\Lambda}_k)^{-1}_{AB} \chi_B } \;.
\ee
Defining a new field $J_A =(C^{\Lambda}_k)^{-1}_{AB} \varphi_B$, we recognize 
\be
\hat{W}_k [\mbK,J ] =W_k [\mbK;\mbC^{\Lambda}_k J ] + \frac{1}{2} J_A (\mbC^{\Lambda}_k)_{AB} J_B \;,
\ee
to be the connected generating functional with local and bilocal sources and covariance $C^{\Lambda}_k$.

Performing a double Legendre transform we define the full 2PI effective action (equation \eqref{eq:GammaG} being the special case $J=\phi=0$):\footnote{We distinguish the various effective actions by their arguments: $\Gamma[\phi,\mbG]$ is the full 2PI effective action, while $\Gamma[\mbG]=\Gamma[\phi=0,\mbG]$ is the reduced action presented in section \ref{sec:model-2PI}; and $\Gamma[\phi]=\Gamma[\phi,\bar\mbG]$ is the 1PI effective action. Lastly, we use the subscript $k$ to denote the presence of the Wilsonian cutoff.}
\be \label{eq:GammaGphi}
\Gamma_{k}[\phi,\mbG]  =  - \hat{W}_{k}[\mbK, J ] +   J_A  \phi_A + \frac{1}{2}  \phi_A \mbK_{AB}  \phi_B + \frac{1}{2} \Tr[G \mbK ] 
\; ,
\ee
where on the right-hand side $J$ and $\mbK$ satisfy $\delta \hat{W}/\delta J_A =\phi_A$ and $\delta \hat{W}/\delta\mbK_{AB} = \frac{1}{2} (\phi_A \phi_B + \mbG_{AB})$.

Acting on $\hat{W}_{k}[\mbK, J ]$ with a $k$-derivative, and using \eqref{eq:Wk}, we obtain:
\be
k \partial_k  \hat{W}_{k}[\mbK, J ] = - \Tr\left[ k \partial_k (\mbC_{k}^{\Lambda})^{-1}  \frac{\delta \hat{W}_{k}}{\delta \mbK}\right] \;.
\ee
Next, acting with a $k$-derivative on \eqref{eq:GammaGphi}, we find that the 2PI effective action satisfies the flow equation:
\be  \label{eq:WE} 
 k \partial_k \Gamma_{k}[\phi,G] =  \frac{1}{2}  \Tr\left[ k \partial_k (C_{k}^{\Lambda})^{-1}  G\right]+ \frac{1}{2} \phi_A k \partial_k (C_{k}^{\Lambda})^{-1}_{AB}  \phi_B \; .
\ee
Going on shell for the two-point function, i.e.\ setting $G=G_k^{\Lambda}$ directly in the flow equation (which is a valid operation since by definition $\delta\Gamma_k/\delta G|_{G=G_k^{\Lambda}} =0$), one recovers the flow equation for the 1PI effective action $\Gamma_{k}[\phi]$ of \cite{Morris:1993qb}.

Expanding the 1PI effective action in powers of the field, and using translation invariance, we write:
\be
\begin{split}
\Gamma_{k}[\phi] & = \sum_{n\geq 0} \frac{1}{n!} \int_{x_1,\ldots, x_n} \Gamma_{k}^{(n)}(x_1,\ldots,x_n) \phi(x_1) \cdots \phi(x_n)\\
&= \sum_{n\geq 0} \frac{(2\pi)^d}{n!} \int_{p_1,\ldots, p_n} \Gamma_{k}^{(n)}(p_1,\ldots,p_n) \phi(p_1) \cdots \phi(p_n) \delta(p_1+\ldots+p_n) \;.\end{split}
\ee

A hierarchy of equations for the $n$-point functions $\Gamma_{k}^{(n)}(x_1,\ldots,x_n)$ is obtained by acting with derivatives on \eqref{eq:WE}, and using the fact that
\be
G_{k}^{\Lambda} = \left(\frac{ \delta^2 \Gamma_{k}[\phi] }{ \delta \phi \delta \phi } \right)^{-1} \;,
\ee
and hence
\be
\frac{ \delta G_{k}^{\Lambda}(x,y)}{ \delta \phi(z) } = - \int_{u,v} G_{k}^{\Lambda}(x,u) \frac{\delta^3 \Gamma_{k}[\phi] }{ \delta \phi(u) \delta \phi(v) \delta \phi(z) } G_{k}^{\Lambda}(v,y) \;.
\ee

Assuming as above that the theory has $\mathbb{Z}_2$ invariance, and thus that the $n$-point functions vanish for odd $n$, we obtain
\be \label{eq:Gamma2-FRG}
k\partial_k \Gamma_{k}^{(2)} = - \frac12  \Tr\left[k\partial_k (C_{k}^{\Lambda})^{-1}  G_{k}^{\Lambda} \Gamma_{k}^{(4)}  G_{k}^{\Lambda}\right]+ k\partial_k (C_{k}^{\Lambda})^{-1} \; ,
\ee
\be \label{eq:Gamma4-FRG}
k\partial_k \Gamma_{k}^{(4)} = - \frac12  \Tr\left[k\partial_k (C_{k}^{\Lambda})^{-1}  G_{k}^{\Lambda} \Gamma_{k}^{(6)}  G_{k}^{\Lambda}\right]+ \sum_{``s,t,u\; {\rm channels}''}   \Tr\left[k\partial_k (C_{k}^{\Lambda})^{-1}  G_{k}^{\Lambda} \Gamma_{k}^{(4)}  G_{k}^{\Lambda} \Gamma_{k}^{(4)}  G_{k}^{\Lambda}\right]\; ,
\ee
and so on.

The second term in \eqref{eq:Gamma2-FRG} is often eliminated by redefining $\Gamma_{k}[\phi] =  \frac{1}{2} \phi_A  (C_{k}^{\Lambda})^{-1}_{AB}  \phi_B + \hat\Gamma_{k}[\phi]$. 
The new two-point vertex can be identified with the self-energy,  $\hat\Gamma_{k}^{(2)}(p) = -\Sigma_k(p)$, and we expect for $p\ll k$ the self-energy to contain a term proportional to the kinetic term of the model: $ \Sigma_k(p) = - (Z_k - 1) p^{2\zeta} + \dots$. 
Therefore, rescaling by $Z_k$ (together with a rescaling of momenta) is necessary in order to restore the covariance $C^{\Lambda}$ in the last line of \eqref{eq:1stRGstep}, as demanded by the second step of the Wilsonian RG.
The rescaling will not play an important role for us, as our couplings are dimensionless, but we should remember to divide $\Gamma_{k}^{(n)}$ by $Z_k^{n/2}$.

For tensor models at large $N$, the hierarchy of equations can be closed\footnote{A closure of this type was considered in \cite{Blaizot:2010zx}, but lacking a melonic $1/N$ expansion, that was the result of an arbitrary truncation rather than a controlled expansion.} as all the $n$-point functions with $n>4$ can be expressed in terms of $ \Gamma^{(2)}_k $ and $ \Gamma^{(4)}_k $, and one could attempt to solve the resulting system of (two) equations. 
In particular, the 1PI four-point vertex in the absence of odd vertices is given by (minus) the amputated connected four-point function  \eqref{eq:4point}, i.e.:
\be \label{eq:4point1PI}
\Gamma^{(4)}{}_{ABEF} = - 2 \left( \frac{{\cal K}}{1 - {\cal K}} {\cal S} \right)_{ AB; E'F' }  G^{-1}_{E'E}  G^{-1}_{F'F} \;.
\ee
which, combined with \eqref{eq:Gamma2-FRG}, gives
\be \label{eq:Gamma2-FRG-redux}
k\partial_k \Gamma_{k}^{(2)} =   \Tr\left[k\partial_k (C_{k}^{\Lambda})^{-1} \left( \frac{{\cal S}}{1 - {\cal K}}  \right)\right] \;.
\ee
Such expression is actually generic, but in the melonic large-$N$ limit the kernel $\cal K$ has the closed expression \eqref{eq:kernel} in terms of the full propagator, rather than being an infinite sum over all kernel diagrams. More importantly $ \Gamma^{(6)}_k $ can be expressed in a closed form in terms of the full propagator and $ \Gamma^{(4)}_k $ (by the type of contact and planar diagrams encountered in \cite{Gross:2017aos}), thus closing equation \eqref{eq:Gamma4-FRG}.

However, such expressions for $ \Gamma^{(4)}_k $ and $ \Gamma^{(6)}_k $ are of limited use here: first, they require renormalization; second, they involve a summation over an infinite series of diagrams.
For these reasons, we will not use explicitly such equations, although it is useful to keep in mind that our construction is implicitly related to them.
First, we will deal with the integrated version of \eqref{eq:Gamma2-FRG-redux}, i.e.\ the Schwinger-Dyson equation $\Gamma_{k}^{(2)} = (C_{k}^{\Lambda})^{-1}-\Sigma[G_k^{\Lambda}]$. Once we have obtained the full renormalized two-point function, we will construct and renormalize the four-point vertex \eqref{eq:4point1PI}, from which we will define the effective couplings, and lastly their beta functions.

Before embarking into that, we should further comment on two non trivial issues raised by the RG framework we just presented.

\paragraph{Wave function renormalization.}

Although the quadratic part of the theory involves the momentum at a non integer power $p^{2\zeta}$, the RG flow generates a wave function renormalization $Z_k$.
For a theory with $\zeta=1$, the wave function is obtained by a 
Taylor expansion in $p$ of the two-point function, but this can not work in our case, $\zeta=d/4<1$, because Taylor expansions generate only integer powers of momenta. 

We will see in section~\ref{sec:2point} how the wave function comes about for non integer powers of momenta.

\paragraph{Subtraction at zero momentum.} Although our theory is massless, we perform a subtraction at zero momentum: the effective dimensionless four-point coupling is 
$ Z_k^{-2}  \Gamma^{(4)}_k \left( 0 ,0 ,0,0 \right)  $. In the Wilsonian picture described above, this is built in. 
 
The issue is subtle. Usually for massless theories one performs the subtraction at a subtraction scale $\mu$, that is, the effective coupling is defined as 
$ Z^{-2}   \Gamma^{(4)}\left( \mu ,\mu ,\mu,\mu \right)   $. 
This is particularly relevant when using dimensional regularization: the subtraction scale $\mu\neq 0$ is required in order to tame the infrared divergences. As it is the only scale in the problem, the RG flow is studied with respect to this subtraction scale. 

 In the Wilsonian picture the infrared divergences are cutoffed by the IR cutoff $k$ and integrating out the modes down to scale $k$ provides an effective theory for the modes with momenta smaller than $k$. Imposing a renormalization condition at some scale $\mu$  does not make sense if the infrared cutoff $k$ is smaller than the subtraction
scale $\mu$: the remaining modes $p<k$ can not reach the scale $\mu$. There is no intrinsic way in which the subtraction scale 
$\mu$ can arise: it can be at most put in by hand, but then the RG map looses its sense for $k<\mu$.

Below, we use the Wilsonian approach and perform the subtraction at zero momentum. This comes at a price: some facts one usually takes for granted when it comes to renormalization need to be revisited.  In particular, one expects that the coefficients of the beta function have finite limits when the UV cutoff is lifted to infinity. These coefficients turn out to be sums over amplitudes of graphs renormalized by the BPHZ subtraction \cite{bogoliubow1957multiplikation,Hepp:1966eg,Zimmermann:1969jj} operator and  the fact that they are finite is simply the statement of the 
BPHZ theorem. However, the  BPHZ theorem does not apply as it only works for massive theories, and Lowenstin's extension \cite{Lowenstein:1974qt,Lowenstein:1975rg} can't be used either as it does not subtract at zero momentum. The zero momentum subtraction in massless theories is much more involved: one needs to use multiscale analysis \cite{Rivasseau:1991ub} and the classification of inclusion forests in order to show that the subtracted amplitudes are indeed finite.

\section{The two-point function}
\label{sec:2point}

We first consider (formally) the theory without cutoffs. The covariance,
on-shell two-point function and on-shell self energy are diagonal in the tensor indices $\mbC_{AB} = \delta_{\mba \mbb} C(x,y)$,
$ \bar \mbG_{AB} = \delta_{\mba \mbb} G(x,y)$ and $\mbS_{AB} = \delta_{\mba \mbb} \Sigma(x,y)$. 
The on-shell Schwinger-Dyson equation becomes at leading and next-to-leading order in $1/N$: 
\be
 \Sigma(x,y) = -m^{2\zeta} \delta_{xy}  - (\lambda_p +\lambda_d) \delta_{xy} G(x,x) + \lambda^2 G(x,y)^3 - 3 \frac{\lambda}{N^{1/2}} \delta_{xy} G(x,x) \;, \qquad G^{-1} = C^{-1} - \Sigma \;.
\ee
Taking $N\to \infty$, this becomes in momentum space:
\be\label{eq:SDE}
\begin{split}
  \Sigma(p) & = - m^{2\zeta}    +     \lambda^2    \int_{q_1,q_2}    G(q_1)  G(q_2)  G( p +q_1 + q_2  )    - ( \lambda_p + \lambda_d ) \int_{q}    G(q) \;,
\crcr
   G(p)^{-1} &  =   C(p)^{-1} -  \Sigma(p) \; .
\end{split}
\ee 
For $C(p)^{-1}=p^2$ (i.e.\ for $\zeta=1$), a simple power counting argument indicates that the solution admits two regimes \cite{Klebanov:2016xxf}: a free scaling regime in the ultraviolet $G(p)^{-1}\sim p^2$ (with $C(p)^{-1}$ dominating over $\Sigma(p)$), and an anomalous scaling regime in the infrared $G(p)^{-1}\sim p^{d/2}$ (with $\Sigma(p)$ dominating over $C(p)^{-1}$).
For the reasons discussed above, we choose $\zeta\neq 1$ to match the infrared conformal behavior.
In fact, with $C(p)^{-1}=p^{2\zeta}$, the Schwinger-Dyson equation is formally solved by $ G(p)^{-1} = Z p^{2\zeta}$ with $\zeta = d/4$:
\be\label{eq:formal}
  Z p^{2\zeta} = p^{2\zeta} +  m^{2\zeta} - \frac{\lambda^2}{Z^3}    \int_{q_1,q_2}   \;\frac{1}{q_1^{2\zeta}} \; \frac{1}{q_2^{2\zeta}} \; \frac{1}{ ( p +q_1 + q_2  )^{2\zeta} } +  \frac{  \lambda_p + \lambda_d }{Z}   \int_{q} \;\frac{1}{ q^{2\zeta} }  \; ,
\ee 
as the double integral (which we call the  melon integral) gives, after a rescaling of $q_1$ and $q_2$ by $|p|$, a global $|p|^{2d - 6\zeta} = |p|^{2\zeta}$. Differently from \cite{Klebanov:2016xxf}, there is only one regime: $\Sigma(p)$ and $C(p)^{-1}$ are of the same order in $p$.
The problem is that both integrals in Eq.~\eqref{eq:formal} are divergent, thus we need regularization and renormalization.

Using the slice propagator\footnote{Thus $\lambda, \lambda_{p},\lambda_{d}$ and $m^{2\zeta}$ become the bare couplings and mass at scale $\Lambda$.} $C^{\Lambda}_k(p) = C(p)\chi^{\Lambda}_k(p)$ and denoting the self energy and the two-point  function with cutoffs $\Sigma^{\Lambda}_k(p)$ and $G^{\Lambda}_k(p)$, the Schwinger-Dyson equation with cutoffs becomes:
\begin{align}\label{eq:SDEcutoff}
 G^{\Lambda}_k(p) & = \frac{1}{ C(p)^{-1} - \chi^{\Lambda}_k(p) \Sigma^{\Lambda}_k(p) } \, \chi^{\Lambda}_k(p)  \equiv
 G\left( p ; \Lambda , k \right) \chi^{\Lambda}_k(p) \;, \\
  G\left( p ; \Lambda , k \right) ^{-1} & =  C(p)^{-1} - \chi^{\Lambda}_k(p) \bigg[ - m^{2\zeta}  -  \left( \lambda_{p} + \lambda_{d}  \right)  \int_{q}  G^{\Lambda}_k(q_1)    +  \lambda^2  \int_{q_1,q_2} G^{\Lambda}_k(q_1) G^{\Lambda}_k(q_2)   G^{\Lambda}_k( p +q_1 + q_2  )   \bigg] \; . \nonumber
\end{align} 
The first equation shows that the two-point function is proportional to the cutoff.

Let us step back once more for a moment and consider again the case $C(p)^{-1}=p^2$: the textbook observation is that at fixed 
$\Lambda$ and $k$, $ G\left( p ;\Lambda,k \right)^{-1}$  is analytic around $p=0$, hence:
\[
 G\left( p ;\Lambda,k \right)^{-1} = G\left( 0 ;\Lambda,k \right)^{-1} +  Z_k p^2 + O(p^4) \;,
\]
and one can extract the wave function renormalization $Z_k$.
Such Taylor expansion (known as the derivative expansion)
has a finite radius of convergence in $p/k$, hence it fails in the $k\to 0$ limit,
but on general grounds (see for example \cite{Delamotte:2007pf}) we expect that:
\begin{itemize}
 \item for $ k \ll p$  the inverse two-point function behaves like $G^{\Lambda}_k(p)^{-1} \sim p^{2-\eta}$,
 \item for $ p \ll k$  the inverse two-point function behaves like $G^{\Lambda}_k(p)^{-1} \sim k^{-\eta} p^2 $,
\end{itemize}
where $\eta$ is the anomalous field dimension. Therefore, in order to extract $\eta$ it is typically enough to obtain the scaling behavior of $Z_k$ with $k$.
However, this is \emph{not} how we are going to treat the two-point function, for the following two reasons. First, as explained before  we are interested in the anomalous scaling of the propagator, to be used in the SDE (we want to do more than the usual RG analysis, we want to show the appearance of the anomalous scaling in the SDE), in the four-point function, and so on. Second, with $C(p)^{-1}=p^{2\zeta}$, $\zeta\neq 1$, we have a non-analytic behavior from the start, and we cannot obtain the wave function renormalization from a Taylor expansion. 

It is unfortunately too difficult to solve the SDE with cutoffs analytically,
therefore we aim to have an ansatz for the two-point function with cutoffs $ G^{\Lambda}_k(p)$ which explicitly exhibits a conformal behavior in the $\Lambda\to \infty, k\to 0$ limit.
We take the ansatz:
\be \label{eq:Gansatz}
 G^{\Lambda}_k(p)=  \frac{1}{Z p^{2\zeta}} \; \chi^{\Lambda}_k(p) \;,
\ee
which reproduces the expected infrared scaling for $k \ll p$ 
with\footnote{The scaling for $p\ll k$ is recovered by observing 
that $\chi^{\infty}_k(p)$ has a series expansion:
\[
\chi^{\infty}_k(p) = \left(\frac{p^2}{k^2}\right)^{\zeta} \frac{1}{\Gamma(\zeta)} \left(\frac{1}{\zeta}  -\frac{1}{1+\zeta} \frac{p^2}{k^2}  +O\left(\left(\frac{p^2}{k^2}\right)^{2}\right) \right) \;,
\]
hence at small $p$ we get:
\[
\hat G^{\infty}_k(p)^{-1} = \frac{1}{p^{2\zeta}} \; \chi^{\infty}_k(p) \simeq
k^{2\zeta} \Gamma(\zeta) \zeta \left(1 +\frac{\zeta}{1+\zeta} \frac{p^2}{k^2}   \right) \; ,
\]
consistent with an anomalous field dimension $\eta = 2 -2\zeta$.} $\eta = 2-2\zeta$.

This ansatz is a solution of the SDE with cutoffs up to terms that are suppressed in the limit $k\to 0$ provided that the mass is tuned to criticality.
To see this, let us denote the cutoffed tadpole and  melon integrals $T^{\Lambda}_k$ and $M^{\Lambda}_k(p)$:
\[
T^{\Lambda}_k = \int_q G^{\Lambda}_k(q)\; , \qquad 
M^{\Lambda}_k(p) =   \int_{q_1,q_2}  G^{\Lambda}_k(q_1)   G^{\Lambda}_k(q_2)     G^{\Lambda}_k( p +q_1 + q_2  ) \; ,
\]
where, using Schwinger parameters, the melon integral writes:
\[
M^{\Lambda}_k(p) = \frac{1}{Z^3(4\pi)^d \Gamma(\zeta)^3 }    \int_{\Lambda^{-2}}^{k^{-2}}  d\alpha_1 d\alpha_2 d\alpha_3  \; 
   \frac{ (\alpha_1\alpha_2 \alpha_3)^{\zeta-1}   }{ ( \alpha_1 \alpha_2 + \alpha_1 \alpha_3 + \alpha_2 \alpha_3)^{d/2} }
   e^{ - p^2 \frac{ \alpha_1 \alpha_2 \alpha_3}{ \alpha_1 \alpha_2 + \alpha_1 \alpha_3 + \alpha_2 \alpha_3 } } \; .
\]
The Schwinger-Dyson equation with cutoffs becomes then:
\begin{equation}
 \label{eq:substSDE}
 Zp^{2\zeta} =  p^{2\zeta} + \bigg[ m_k^{2\zeta}    - \lambda^2    \big( M^{\Lambda}_k(p) - M^{\Lambda}_k(0) \big) \bigg] \chi^{\Lambda}_k(p) \;, 
\end{equation}
where all the $p$-independent contributions in the square bracket have been absorbed in the renormalized mass
\be
m_k^{2\zeta} =  m^{2\zeta}  +     \left( \lambda_{p} + \lambda_{d}  \right) T^{\Lambda}_k  - 
  \lambda^2  M^{\Lambda}_k(0)\;.
\ee 
We can tune the UV mass so as to both cancel the ultraviolet mass divergences in the SDE for the two-point function
 and ensure that the renormalized mass goes to zero in the $k=0$ limit. 
 
 \begin{proposition}
There exists $m^{2\zeta}$ depending only on $\Lambda$ and the bare coupling constants $ \lambda_{p} , \lambda_{d} , \lambda$ such that 
\[
 \lim_{k\to 0}m_k^{2\zeta} =0 \;.
\]
\end{proposition}

\begin{proof}
 
 The tadpole integral is:
 \[
 \begin{split}
  T^{\Lambda}_k  = \int_{q}   G^{\Lambda}_k(q) & =  \frac{1}{Z} \int_{\Lambda^{-2}}^{k^{-2}}  d\alpha  \int_q  \; \frac{ \alpha^{\zeta-1}}{\Gamma(\zeta)} e^{ -\alpha q^2} = \frac{1}{ Z  (4\pi)^{d/2} }\; \frac{-4}{\Gamma(d/4)d} \left( (k^2)^{\frac{d}{4}} - (\Lambda^2)^{\frac{d}{4}} \right) \;,
 \end{split}
 \]
 where we used $\zeta = \frac{d}{4}$. Combining this with the $M^{\Lambda}_k(0)$ integral computed in Appendix \ref{app:melon} we get
 $m_k^{2\zeta} = m^{2\zeta} + A(\lambda_{p},\lambda_{d},\lambda) ( \Lambda^{d/2} 
 - k^{d/2} ) $ with:
 \[
A(\lambda_{p},\lambda_{d},\lambda)=\frac{4\left( \lambda_{p } + \lambda_{d}  \right)}{Z(4\pi)^{d/2}\Gamma(d/4)d}-\lambda^2\frac{24}{dZ^3(4\pi)^d\Gamma(d/4)^3}\int_1^{\infty}dx \int_1^{\infty}dy \frac{x^{-1}y^{d/4-1}}{\left(1+y+xy\right)^{d/2}} \;,
 \]
Choosing 
\be \label{eq:baremass}
m^{2\zeta} = - A(\lambda_{p},\lambda_{d},\lambda) \Lambda^{d/2} \;,
\ee
we obtain:
\[
 m^{2\zeta}_k =- k^{d/2} A\left( \lambda_p, \lambda_d, \lambda \right) \;,
\]
which goes to 0 when sending $k\to 0$. 
\end{proof}

 \begin{proposition}
Choosing $m^{2\zeta}$ as in \eqref{eq:baremass}, the Schwinger-Dyson equation \eqref{eq:SDEcutoff} is solved by the ansatz \eqref{eq:Gansatz}, with $\zeta= \frac{d}{4}$, and with $Z$ satisfying
\be\label{eq:wave}
Z^4 - Z^3  =   \lambda^2 \frac{1}{(4\pi)^d } \;\frac{\Gamma \left( 1 -\frac{d}{4} \right) }{ \frac{d}{4}\Gamma\left( 3 \frac{d}{4}\right)} \;.
\ee
\end{proposition}

\begin{proof}

We show in Appendix \ref{app:melon} that:
\[
M^{\Lambda}_k(p)
   =M^{\Lambda}_k(0)- \frac{p^{2d-6\zeta} f\left( \frac{ k^2 }{p^2}, \frac{p^2}{\Lambda^2} \right)}{Z^3(4\pi)^d \Gamma(\zeta)^3} \;, \qquad
 f(0,0)  =\frac{\Gamma(1-d+3\zeta) \Gamma(d/2-\zeta)^3 }{( d-3\zeta ) \Gamma(3d/2-3\zeta)}\;.
\]
Choosing $m^{2\zeta}$ as in  \eqref{eq:baremass} to exactly cancel the UV divergent pieces arising from $T^{\Lambda}_k $ and $M^{\Lambda}_k (0)$,
we obtain a renormalized mass $m^{2\zeta}_k $ which is tuned to criticality
$\lim_{k\to 0} m_k^{2\zeta} = 0 $. With this choice, 
we can take $\Lambda \to \infty, k\to 0$ in Eq.~\eqref{eq:substSDE} and obtain:
\[
\begin{split}
 (Z-1) p^{2\zeta} =    \lambda^2p^{2d-6\zeta} \frac{ f(0,0) }{Z^{3}(4\pi)^{d} \Gamma(\zeta)^{3} }  \; ,
\end{split}
\]
which is solved by $\zeta= \frac{d}{4}$ and
\[
Z^4 - Z^3  =   \lambda^2 \frac{1}{(4\pi)^d } \;\frac{\Gamma \left( 1 -\frac{d}{4} \right) }{ \frac{d}{4}\Gamma\left( 3 \frac{d}{4}\right)} \;.
\]
\end{proof}

Notice that dropping the $Z^3$ term in \eqref{eq:wave} (which comes from the inverse free covariance) we recover the result of \cite{Klebanov:2016xxf}.

\section{The four-point couplings}
\label{sec:4point}
 
We denote $\hat \delta^{p}_{\mba \mbb ; \mbc \mbd} = \frac{1}{N^2} \delta^{p}_{\mba \mbb ; \mbc \mbd}$ and $\hat \delta^d_{\mba \mbb ; \mbc \mbd} = \frac{1}{N^3} \delta^{d}_{\mba \mbb ; \mbc \mbd}$
the rescaled pillow and double-trace contraction operators. The four-point function:
\begin{align}\label{eq:2.0}
 \braket{\varphi_{A}  \varphi_{B} \varphi_{C} \varphi_{D} }^c  =2 \left( \frac{{\cal K}}{1 - {\cal K}} {\cal S} \right)_{ AB; C'D' }  G_{C'C}  G_{D'D}  \; ,
\end{align}
is computed in terms of the four-point kernel ${\cal K}$ which at leading and next-to-leading order in $1/N$ is, using the shorthand notation $G_{xy} = G(x,y)$:
\begin{equation}
\begin{split}
  {\cal K}_{ (\mba,x')(\mbb,y') ; (\mbc,z)(\mbd,w) } = & G_{x'x} G_{y'y}  \bigg[  
    -  \lambda_p   \delta_{xy} \delta_{xz} \delta_{xw} \hat \delta^p_{\mba\mbb ; \mbc\mbd} -  \lambda_d  \delta_{xy} \delta_{xz} \delta_{xw} \hat \delta^d_{\mba\mbb ; \mbc\mbd}   
    + 3 \lambda^2 G_{xy}^2 \delta_{xz}\delta_{yw} \hat \delta^p_{\mba\mbb ; \mbc\mbd} \crcr
    & -   \frac{\lambda}{N^{3/2}} \delta_{xy}  \delta_{xz} \delta_{xw}   \left( \frac{  \delta^t_{\mba\mbb   \mbc\mbd} +\delta^t_{\mba\mbb   \mbd\mbc} +  \delta^t_{\mba  \mbc \mbb  \mbd} +  \delta^t_{\mba  \mbd \mbb  \mbc} 
    +  \delta^t_{\mba \mbc \mbd   \mbb} + \delta^t_{\mba \mbd \mbc   \mbb} }{2} \right)
 \bigg] \; .
\end{split} 
\end{equation}

Due to the $O(N)^3$ symmetry as well as the color permutation symmetry, minus the amputated four-point function
$\Gamma^{(4)}_{ (\mba,x) (\mbb,y)  (\mbc,z) (\mbd,w) } $ is a sum of three classes of terms.

\paragraph{Tetrahedral terms.} We have six tetrahedral terms:
    \[
   \sum_{ (\mbb,y) (\mbc,z) (\mbd,w) }
      \delta^{t}_{\mba \mbb\mbc\mbd} \Gamma^{(4,t)}_{xyzw}  \;, 
     \qquad \Gamma^{(4,t)}_{xyzw}  =  \frac{  \lambda }{4 N^{3/2}} \delta_{xy} \delta_{xz} \delta_{xw} \;,
    \]
where the sum runs over the six permutations of the couples $(\mbb,y) (\mbc,z) (\mbd,w) $, coming from the next-to-leading order contribution to the kernel.
The effective tetrahedral coupling is then exactly 
\be\label{eq:tetra}
  g =k^{d-4\zeta} \frac{\lambda } {Z^2} \;.
\ee
 
\paragraph{Pillow and double-trace terms.} We have three pillow and three double-trace terms, corresponding to the three channels $(\mba,x) (\mbb,y)  \to (\mbc,z) (\mbd,w) $,
$(\mba,x) (\mbc,z)   \to (\mbb,y) (\mbd,w) $ and $(\mba,x) (\mbd,w)  \to (\mbb,y)   (\mbc,z)  $. We write minus the pillow 1PI four-point function as:
    \[
       2  \left( \hat \delta^{p}_{\mba\mbb; \mbc\mbd} \Gamma^{(4,p)}_{xy; zw} +   \hat \delta^{p}_{\mba\mbc; \mbb\mbd} \Gamma^{(4,p)}_{xz; yw}  +   \hat \delta^{p}_{\mba\mbd; \mbb\mbc} \Gamma^{(4,p)}_{xw; yz} \right) \;,
    \]
where the factor 2 is conventional. The double-trace contribution is obtained by changing the superscript $p$ to $d$. At leading order in $1/N$ the sum of the pillow and double-trace contributions in one channel is:
\be
\begin{split}
   -  \hat \delta^{p}_{\mba\mbb; \mbc\mbd} \Gamma^{(4,p)}_{xy; zw}  -  \hat \delta^{d}_{\mba\mbb; \mbc\mbd} \Gamma^{(4,d)}_{xy; zw}    =
      G^{-1}_{xx'} G^{-1}_{yy'}  \left( \frac{K}{1-K} \right)_{x'y';zw} \;,
\end{split}
\ee
where $K$ is the on shell leading order four-point kernel:
\[
 K_{ (\mba,x')(\mbb,y') ; (\mbc,z)(\mbd,w) } =   G_{x'x} G_{y'y}  \bigg[  
    -  \lambda_p   \delta_{xy} \delta_{xz} \delta_{xw} \hat \delta^p_{\mba\mbb ; \mbc\mbd} -  \lambda_d  \delta_{xy} \delta_{xz} \delta_{xw} \hat \delta^d_{\mba\mbb ; \mbc\mbd}   
    + 3 \lambda^2 G_{xy}^2 \delta_{xz}\delta_{yw} \hat \delta^p_{\mba\mbb ; \mbc\mbd} 
 \bigg] \; .
\]
In momentum space, we have $\Gamma^{(4,d) }_{p_1p_2p_3p_4} = (2\pi)^d \delta(p_1+p_2+p_3+p_4 )\Gamma^{(4,d)}(p_1,p_2,p_3,-p_1-p_2-p_3)$,
and the four-point function in the channel $\mba \mbb  \to \mbc \mbd $ is:
\be\label{eq:fourpointdef}
  - \hat \delta^{p}_{\mba\mbb; \mbc\mbd} \Gamma^{(4,p)}_{p_1p_2;  r_1   r_2  } - \hat \delta^{d}_{\mba\mbb; \mbc\mbd } \Gamma^{(4,d)}_{p_1p_2;  r_1  r_2  } 
 =   \frac{1}{  G(p_1)  G(p_2)   } \left(  \frac{K}{1-K} \right)_{\mba\mbb; \mbc\mbd}( p_1, p_2 ; r_1 , r_2 ) \; ,
\ee
with $1$ the identity operator on bilocal functions $1 = (2\pi)^{2d} \delta(p_1-q_1) \delta(p_2-q_2)$ and the four-point kernel $K$ in momentum space:
\be
\begin{split}\label{eq:fourpointkernel}
  K_{ p_1,p_2; q_1,q_2 } = (2\pi)^d \delta(p_1+p_2 - q_1 -q_2)   G(p_1) G(p_2)  
 \bigg[ \hat \delta^{p }  \; 3  \lambda^2 \int_q  G(q) G(q + p_1 -q_1)     -   \hat \delta^{p }  \lambda_{p}   - \hat \delta^{d }  \lambda_{d}  \bigg] \; .
\end{split}
\ee

When expanding the geometric series in $K$ we need to deal with powers of $\hat \delta^p$ and $\hat \delta^d$. This is slightly unpleasant as $\hat \delta^p \hat \delta^p =  \frac{1}{3} \hat \delta^p + \frac{2}{3} \hat \delta^d$, 
$\hat \delta^p \hat \delta^d =  \hat \delta^d$ and $\hat \delta^d \hat \delta^d = \hat \delta^d$, that is $\hat \delta^p$ and 
$\hat \delta^d$ are not mutually orthogonal. In particular, the pillow and double-trace couplings mix. It is convenient to parameterize the interaction in terms of two independent couplings which do not mix. The operators:
\be
 P_1 = 3( \hat \delta^p - \hat \delta^d ) \; ,\qquad P_2 = \hat \delta^d \;,
\ee
are two mutually orthogonal projectors which span the interaction space.\footnote{This corresponds to the traceless-trace decomposition in the intermediate field representation of the pillow and double-trace interactions \cite{Benedetti:2018ghn}.} We parametrize the interaction in terms of  
$\lambda_1 = \frac{\lambda_p}{3}$ and $\lambda_2 = \lambda_d + \lambda_p$.
Thus, the four-point kernel in momentum space becomes:
\begin{align}\label{eq:newker}
   K_{ p_1,p_2; q_1,q_2 } = & (2\pi)^d \delta(p_1+p_2 - q_1 -q_2)   G(p_1) G(p_2)  \crcr
&\qquad \bigg[ \bigg( \lambda^2 \int_q G(q)G(q+p_1-q_1)       -\lambda_1 \bigg) P_1 + \bigg( 3\lambda^2 \int_q G(q)G(q+p_1-q_1) -  \lambda_2 \bigg) P_2 
\bigg]\; ,
\end{align}
and the effective four-point function (where 
$\Gamma^{(4;1)}  = \frac{1}{3}\Gamma^{(4,p)} $ and $\Gamma^{(4;2)}  = 
\Gamma^{(4,d)} +  \Gamma^{(4,p)} $):
\be\label{eq:reparam}
- \Gamma^{(4;1)} P_1 - \Gamma^{(4;2)} P_2 = G^{-1}G^{-1} \frac{K}{1-K} \; .
\ee

\subsection{The bare expansion}

Eq.~\eqref{eq:reparam} and \eqref{eq:newker} are used to obtain the bare expansion of the running couplings:
\be
g_{1} = k^{d-4\zeta} \frac{\Gamma_k^{(4;1)}(0,0,0,0)}{Z^2} \;,\qquad
g_{2} = k^{d-4\zeta} \frac{\Gamma_k^{(4;2)}(0,0,0,0)}{Z^2}  \; ,
\ee 
as decoupled series in the bare couplings. The two cases are identical, up to replacing $\lambda^2$ by $3\lambda^2$, hence we discuss below only $g_1$. 
\begin{figure}[ht]
\begin{center}
\includegraphics[scale=1.3]{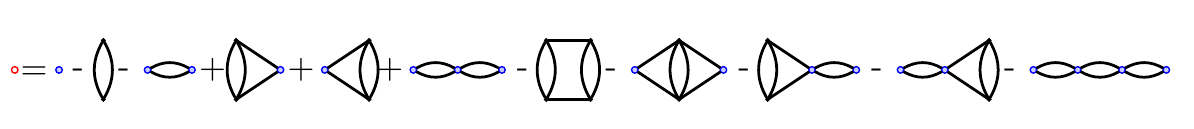} 
 \caption{The bare series up to quartic order. For $g_1$ the blue vertices represent $\lambda_1$ and the rungs contribute $\lambda^2$, while for $g_2$
 the blue vertices represent $\lambda_2$ and the rungs 
 contribute $3\lambda^2$.} \label{fig:bare1}
 \end{center}
\end{figure}

The bare series is a sum over connected amputated chain graphs $\cG$ depicted in Fig.~\ref{fig:bare1}.  A connected amputated chain graph $\cG$ is a sequence of irreducible pieces connected one to another by pairs of parallel horizontal edges. The irreducible pieces are either vertical ladder rungs with two tetrahedral couplings or bare vertices $\lambda_1$. There are $2^{n}$ chain graphs with $n$ irreducible parts (that is vertical rungs or bare vertices). To each graph we associate an
\emph{amplitude}:
\be\label{eq:ampli1}
 A(\cG) = \int_{\Lambda^{-2}}^{k^{-2}} 
 \left( \prod_{e\in \cG} d\alpha_e \;  \alpha_e^{\zeta-1}\right) \; 
 \frac{ 1 }{ \big[ \sum_{ {\cal T} \subset \cG} \prod_{e\notin {\cal T} } \alpha_e \big]^{d/2} } \;,
\ee
where $e\in \cG$ denotes the edges of $\cG$ and ${\cal T}$ runs over the spanning trees in $\cG$ (see for example \cite{Rivasseau:1991ub,Krajewski:2008fa}).
For instance, the amplitudes of the graphs depicted on the right hand side in Fig.~\ref{fig:bare1} are, from left to right:
\[
1 \, , \;\;  U_1 \, , \;\;  T_0 \, , \;\;  S_1 \, , \;\; S_1 \, , \;\; T_0^2 \, , \;\; 
U_2 \, , \;\; T_1 \, , \;\; S_1T_0\, , \;\;  T_0S_1 \, , \;\; T_0^3 \;,
\]
where $S_1$, $T_0,T_1$ and $U_1,U_2$ denote the integrals (to simplify the notation we suppress the measure):
\begin{align}\label{eq:SDTU}
 U_1 = & T_0 =   \int_{\Lambda^{-2}}^{k^{-2}} \frac{(a_1 a_2)^{\zeta-1}}{ (a_1 + a_2)^{d/2}} \equiv D  \; , \qquad
 S_1 =  \int_{\Lambda^{-2}}^{k^{-2}} \frac{(a_1 a_2 b_1 b_2)^{\zeta-1}}{ \big[(a_1 + a_2)(b_1 + b_2 ) + a_1 a_2 \big] ^{d/2}}  \;,  \crcr
 U_2 = & \int_{\Lambda^{-2}}^{k^{-2}} \frac{  (a_1 a_2 b_1 b_2 c_1  c_2)^{\zeta-1}   }
  { \big[ (a_1 + a_2)(b_1 + b_2) (c_1 + c_2) + (a_1 + a_2) c_1c_2 + a_1a_2(c_1+c_2)\big]^{d/2}  }  \;, \crcr
 T_1 = & \int_{\Lambda^{-2}}^{k^{-2}} \frac{  (a_1 a_2 b_1 b_2 c_1  c_2)^{\zeta-1}   }
  { \big[ (a_1 + a_2)(b_1 + b_2) (c_1 + c_2) +  b_1b_2(a_1+a_2+c_1+c_2) \big]^{d/2}  }  \;.
\end{align}
 Observe that (setting again $\zeta=d/4$) the amplitude of a graph diverges like some power of 
$ \ln (\Lambda^2/ k^2) $. 
 
The chain graph consisting in a bare vertex has amplitude $1$. We denote 
$\mathfrak{G}$ the set of connected chain graphs with at least two internal vertices. The number of tetrahedral vertices of $\cG\in\mathfrak{G}$, $n_t(\cG)$, is always even. We denote $n_1(\cG)$ the numbers of blue vertices of $\cG$. The graphs $\cG\in\mathfrak{G}$ are such that $n_t(\cG)+n_1(\cG) \ge 2$.
We rescale the bare and effective coupling as 
$\tilde g_1 = (4\pi)^{ - d/2} \Gamma(\zeta)^{-2} g_1$ and so  on. Forgetting the tilde, and recalling that $g = \lambda/Z^2$ the bare expansion writes: 
\begin{align*}
  g_1(\lambda_1, g)  =   \frac{\lambda_1}{Z^2}  
  +  \sum_{\cG\in \mathfrak{G}} (-1)^{1 + n_1(\cG) } 
      g^{n_{t}(\cG)} 
 \bigg( \frac{ \lambda_1}{Z^2}  \bigg)^{n_1(\cG)} 
 A(\cG)\; .
\end{align*}
The same formula holds for $g_2$ by replacing $g^2$ by $3g^2$.
Up to total degree $4$ in the coupling constants the bare expansion is:
\be\label{eq:Bare}
\begin{split}
  g_1(\lambda_1,g)  = &  \bigg(\frac{ \lambda_1}{Z^2}\bigg)- g^2 U_1  - 
  \bigg(\frac{ \lambda_1}{Z^2}\bigg)^2T_0 + 2g^2 \bigg(\frac{ \lambda_1}{Z^2}\bigg) S_1 + \bigg(\frac{ \lambda_1}{Z^2}\bigg)^3  T_0^2 \crcr
  &\qquad  
   -g^4 U_2 -g^2 \bigg(\frac{ \lambda_1}{Z^2}\bigg)^2 T_1 - 2g^2 \bigg(\frac{ \lambda_1}{Z^2}\bigg)^2 S_1 T_0 - \bigg(\frac{ \lambda_1}{Z^2}\bigg)^4 T_0^3 \;.
\end{split}
\ee 

 \paragraph{One vertex reducible graphs.}
\begin{figure}[ht]
\begin{center}
 \includegraphics[scale=1]{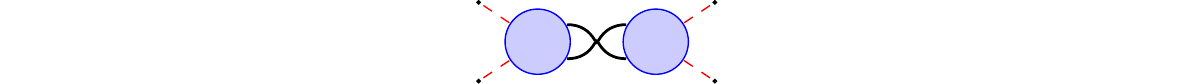} 
 \caption{One vertex reducible graph. We represented in dashed red the amputated external edges.} \label{fig:graphs3}
 \end{center}
\end{figure}

A graph is called one vertex reducible (1VR) if it disconnects into two nontrivial graphs (i.e graphs having internal vertices) by cutting a vertex (see Fig.~\ref{fig:graphs3}).
As there are no two-point subgraphs (we are using the full propagator), any 1VR four-point graph disconnects into two four-point graphs by cutting the vertex ``vertically'' and adding a pair of external edges on each resulting ``half vertex''. We write $\cG = \cG_1 \cG_2$. 
By this procedure, any four-point graph can be decomposed as a chain $\cG = \cG_1\dots \cG_q$ where $\cG_i$ are one vertex irreducible (1VI). The amplitude factors on the 1VI components : $A(\cG) = A(\cG_1)\dots A(\cG_q)$. We classify the 1VI graphs into three families:
\begin{itemize}
 \item The pure \emph{ladders} depicted in Fig.~\ref{fig:pureladders} consisting in a nonempty sequence of vertical ladders with tetrahedral vertices. We denote $U_r$ the amplitude of the ladder graph with $r$ rungs (by some abuse of notation we will denote the graph itself also $U_r$ when no confusion can arise). The sum over the ladders is:
 \be
    U(g) =  \sum_{r\ge 1} g^{2r} U_r  = g^2 D + 
  g^4 U_2 + \dots \;.
 \ee
\begin{figure}[H]
\begin{center}
\includegraphics[scale=1]{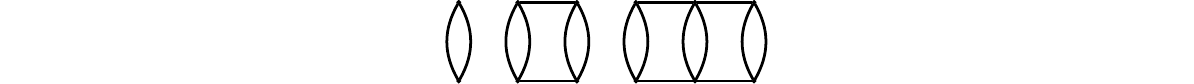} 
 \caption{The pure ladders $U_1,U_2$ and $U_3$.} \label{fig:pureladders}
 \end{center}
\end{figure}
 
 \item The ``v-ladders'' or \emph{caps}, that is ladders having a blue bare vertex at one end, depicted in Fig.~\ref{fig:pladders}. They consist in a blue bare vertex followed by a nonempty sequence of vertical ladder rungs with tetrahedral vertices. 
 We denote $S_r$ the amplitude of the cap with $r$ rungs (and the graph itself also $S_r$). The sum over the caps is:
 \be 
 S(g) =  \sum_{r \ge 1 }g^{2r} S_r 
        = g^2 S_1 + \dots \;.
 \ee
\begin{figure}[H]
\begin{center}
\includegraphics[scale=1]{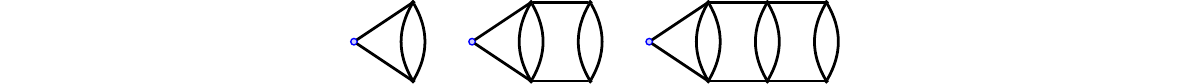} 
 \caption{The caps $S_1,S_2$ and $S_3$.} \label{fig:pladders}
 \end{center}
\end{figure}
 \item The ``vv-ladders'' or \emph{double caps} having a blue bare vertex at each end depicted in Fig.~\ref{fig:pladdersp}. They consist in a blue bare vertex followed by a \emph{possibly empty} sequence of vertical ladder rungs with tetrahedral vertices followed by a blue bare vertex. 
   We denote $T_r$ the amplitude of the double cap graph with $r$ rungs (and the graph itself also $T_r$). The sum over the double caps is:
   \be
  T (g) =  \sum_{r\ge 0} g^{2r } T_r
      = D + g^2 T_1 + \dots \;,
  \ee
\begin{figure}[H]
\begin{center}
\includegraphics[scale=1]{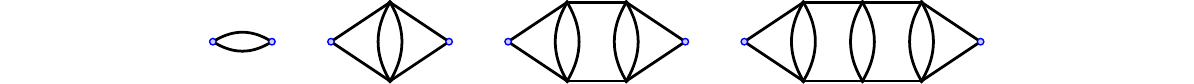} 
 \caption{Some double caps.} \label{fig:pladdersp}
 \end{center}
\end{figure}
\end{itemize}
Observe that in the generating functions $S(g)$ and $T(g)$ we have \emph{not} included any coupling constants for the blue vertices.
By counting the number of reducibility vertices in a graph the bare series is simply:
\begin{align}\label{eq:bare1VR}
 g_1 =&  -U(g) + \frac{\lambda_1}{Z^2} \big[ 1 + S(g) \big] \bigg[ \sum_{q\ge 0}
 \bigg(- \frac{ \lambda_1}{Z^2} \bigg)^q \sum_{r_1,\dots r_{q} \ge 0 }  \prod_{i=1}^q \lambda^{2r_i} T_{r_i} \bigg]
 \big[ 1 + S(g)\big] \crcr 
= &- U(g) + \bigg( \frac{\lambda_1}{Z^2} \bigg) 
\frac{\big[ 1+S(g) \big]^2}
 { 1 + \frac{ \lambda_1 }{Z^2} T( g ) } \;. 
\end{align}

A similar expression holds for the usual $\varphi^4$ model (with of course other types of graphs contributing too): $U,S,T$ can be computed directly starting from the four-point kernel, by separating the contribution of the bare vertex from the rest. The important difference is that in general $U,S$ and $T$ depend on $\lambda_1$, whereas in our case they depend only on the parameter $g$. This makes the $\beta$ function in our case particularly simple: as we will see below, the all orders $\beta$ function is only quadratic in the running coupling.

\paragraph{The renormalized expansion.}
The bare series can be inverted to yield the renormalized series. 
This is usually done by iterative substitutions:
\be
\frac{ \lambda_1 }{Z^2} = g_1 + \bigg[ \sum_{\cG\in \mathfrak{G}} (-1)^{n_1(\cG) } g^{n_t(\cG)} \bigg(\frac{ \lambda_1}{Z^2} \bigg)^{n_1(\cG)}A(\cG)\bigg]_{ \frac{\lambda_1}{Z^2} = 
g_1 + \sum_{\cG\in \mathfrak{G}} (-1)^{n_1(\cG) } g^{n_t(\cG)}
\big(\frac{\lambda_1}{Z^2} \big)^{n_1(\cG)}A(\cG)  } \; ,
\ee
as represented in the Figure \ref{fig:graphs2}. 
\begin{figure}[ht]
\begin{center}
\includegraphics[scale=1.3]{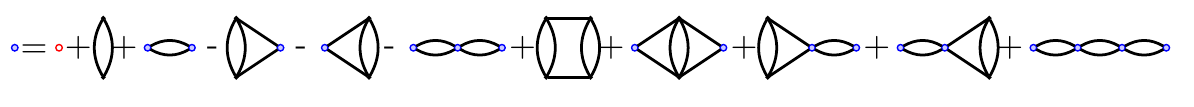} 
 \caption{The renormalized series by iterated substitutions. On the right hand side the 
 vertices $\lambda_1$ (in blue) should be iteratively substituted with the right hand side itself.} \label{fig:graphs2}
 \end{center}
\end{figure}

In the most general case the result of iterated insertions can be written compactly using Zimmermann forests \cite{Zimmermann:1969jj}.
However, in our case we can invert the bare series directly:
 \be\label{eq:renodirect}
  \frac{\lambda_1}{Z^2} = 
  \frac{g_1+U(g)}{ \big[1+S(g) \big]^2 - U(g)T(g) - g_1 T(g)}  \;.
\ee
For the usual $\varphi^4$, the 1VI expansion can not be inverted so simply:
while in our case Eq.~\eqref{eq:renodirect} is a series in $g_1$ whose coefficients are series in the parameter $g$, in the usual $\varphi^4$ the coefficients themselves depend on $g_1$.
The renormalized expansion in our case up to quartic order is:
\begin{align*}
 \frac{ \lambda_1}{Z^2} = & g_1 +  g^2 U_1 + g_1^2 T_0  -
 2 g^2 g_1 (S_1-U_1T_0)  
              +  g_1^3  T_0^2 +    g_1^4  T_0^3 \crcr
          &    + g^4 \big[ U_2 - 2 S_1 U_1  + U_1^2T_0 \big] 
              + g_1^2 g^2 \big[ T_1 - 4S_1T_0 + 3U_1T_0^2  \big] \; .
\end{align*}

\subsection{Beta functions} 
\label{sec:4point-beta}

Let us denote $ \pmb{\partial} = k\partial_k $.
The $\beta$ function can be computed in two ways. Usually one starts from the renormalized series $\lambda_1(k,g_1)$, where $k$ denotes the explicit dependence of the coefficients on the IR cutoff and  equates to zero the scale derivative of the bare coupling. This has the advantage that the $\beta$ function is directly written in terms of the running
coupling. In our case however, it is more convenient to derive directly the bare expansion in Eq.~\eqref{eq:bare1VR} and rewrite the derivative as:
\be\label{eq:beta}
 \beta_{g_1} = \pmb{\partial} g_1
 = \beta_0^{g}
  - 2  \beta_1^{g} g_1  + \beta_2^{g} g_1^2\;,
\ee
with the coefficients of the $\beta$ function given by:
\begin{align}\label{eq:coef}
 \beta_0^{g}  &= - \pmb{\partial} U + 2 \frac{U}{1+S} \pmb{\partial} S - \frac{U^2}{(1+S)^2} \pmb{\partial} T  
 = - g^2 \pmb{\partial} U_1 -
        g^4 \big[ \pmb{\partial}  U_2 - 2 U_1 \pmb{\partial} S_1 + U_1^2 \pmb{\partial} T_0 \big]  + O(g^6)
 \; ,\crcr
 \beta_1^{g} &= - \frac{1}{1+S} \pmb{\partial}S + \frac{U}{(1+S)^2} \pmb{\partial} T = -
 g^2 ( \pmb{\partial} S_1 - U_1 \pmb{\partial} T_0 ) 
     + O(g^4)
 \;, \\ \nonumber
 \beta_2^{g} &= - \frac{1}{(1+S)^2} \pmb{\partial} T
 = - \pmb{\partial} T_0 - g^2 ( \pmb{\partial} T_1 -2
          S_1 \pmb{\partial} T_0 ) + O(g^4)  \; .
\end{align}
We will discuss these coefficients further in section \ref{sec:divergences}.
At three loops the $\beta$ function is:
\be\label{eq:beta3loops}
 \beta_{g_1} = - g^2 \pmb{\partial} U_1 -  g^4 
 \big[ \pmb{\partial}  U_2 - 2  U_1 \pmb{\partial} S_1 + U_1^2 \pmb{\partial} T_0 \big]
 + 2g_1 \bigg[ g^2 ( \pmb{\partial} S_1 - U_1 \pmb{\partial} T_0 )\bigg]- g_1^2 \bigg[\pmb{\partial} T_0 + g^2 
 ( \pmb{\partial} T_1 -2 S_1 \pmb{\partial} T_0 ) \bigg] \;,
\ee
The $\beta$ functions of the original pillow and double-trace couplings $\lambda_p = 3\lambda_1,\lambda_d = \lambda_2 - 3\lambda_1$ can be reconstructed as: 
\begin{align*}
  \beta_{p} =&  3 \beta_0^{g} + 2\beta_1^{g} \, g_p + \frac{1}{3} 
  \beta_2^{g}\, g_p^2  \;, \crcr
\beta_{d} =&   \big[ \beta_0^{\sqrt{3}g}  - 3 \beta_0^{g} \big]
 + 2 \beta_1^{\sqrt{3}g} \; g_d + 2 
 \big[\beta_1^{\sqrt{3}g} - \beta_1^{g} \big] g_p  + 
 \beta_2^{ \sqrt{3}g }  \, g_d ^2 
+ 2 \beta_2^{ \sqrt{3}g } g_d  g_p +
 \big[ \beta_2^{\sqrt{3}g} - \frac{1}{3}  \beta_2^{ g} \big] g_p^2  
 \; ,
\end{align*}
which is at two loops, using $U_1 = T_0 = D$:
\begin{align*}
 \beta_{p} = & - 3g^2 \pmb{\partial} D + 2 g^2 ( \pmb{\partial} S_1 - D \pmb{\partial} D ) g_p - \frac{1}{3}\pmb{\partial} D g_p^2  \; , \crcr
 \beta_d = & 6 g^2 ( \pmb{\partial} S_1 - D \pmb{\partial} D ) g_d
  + 4 g^2 ( \pmb{\partial} S_1 - D \pmb{\partial} D ) g_p
   - \bigg(  g_p^2 + 2g_pg_d + \frac{2}{3} g_p^2\bigg)\pmb{\partial} D \;.
\end{align*} 
Notice that the $g_p$- and $g-d$-independent part of $ \beta_d$ starts at order $g^4$, as expected from the minimal resolution of the double-trace bubble in terms of tetrahedra (see Fig.~\ref{fig:graph}).

Lastly we can explicitly check that the lowest order coefficients of $\beta_{0,1,2}^{g}$ are convergent and given by:
\be
 \pmb{\partial} D = -4 \int_{k^2\Lambda^{-2}}^{1}
 d\alpha \;\frac{\alpha^{\zeta-1}}{(1+\alpha)^{d/2}} \to_{\Lambda\to\infty}
    - 2 \frac{\Gamma\left( \frac{d}{4}\right)^2}{ \Gamma \left( \frac{d}{2}\right)}  \;,
\ee
\be
 -  \frac{1}{4} (\pmb{\partial S} - D \pmb{\partial}D ) 
  = I_1 +I_2 \;,
\ee
with
\begin{align*}
 I_1 & = \int_0^1 \frac{ (a_1a_2 b_2)^{\zeta-1}}{[ (a_1+a_2)(1+b_2) + b_2]^{d/2}}  \; ,\crcr
 I_2 & = -\frac{d}{2} \int_0^1 du \int_0^1 
 \frac{ (a_2 b_1b_2)^{\zeta-1} b_1 b_2}{[ (1+a_2)(b_1+b_2) + u b_1b_2]^{d/2+1}} \; .
\end{align*}

\subsection{Flow and fixed points} 
\label{sec:4point-FP}

Being quadratic in $g_1$, the beta function \eqref{eq:beta} admits  two fixed points:
 \begin{align}
 g_{1\pm} =  \frac{
    \beta_1^{g} 
    \pm \sqrt{ (\beta_1^{g})^2 -\beta_0^{g}\beta_2^{g} }
    }{\beta_2^{g}}  = \pm\sqrt{-g^2} + O(g^2) \; .
   \end{align}  
In fact, we can even solve the full flow, in terms of the beta function coefficients \eqref{eq:coef}.
Taking $g$ to be purely imaginary and small, so that $(\beta_1^{g})^2 -\beta_0^{g}\beta_2^{g} >0$, we find
\be \label{eq:flow}
g_1(k) = \frac{
    \beta_1^{g} 
    - \sqrt{ (\beta_1^{g})^2 -\beta_0^{g}\beta_2^{g} } \,\tanh\left( \sqrt{ (\beta_1^{g})^2 -\beta_0^{g}\beta_2^{g} } \ln(k/k_0) + c  \right)
    }{\beta_2^{g}} \;,
\ee
where $c$ is an integration constant to be fixed by the initial condition $g_1(k_0)=\bar g_1$:
\be
c = {\rm arctanh}\left( \frac{\beta_1^{g} - \beta_2^{g} \bar g_1 }{ \sqrt{ (\beta_1^{g})^2 -\beta_0^{g}\beta_2^{g} } }  \right) \;.
\ee
We then see that $g_{1+}$ is an IR fixed point (reached for $k\to 0$) and $g_{1-}$ is a UV fixed point (reached for $k\to \infty$).

The corresponding critical exponents are:
   \be
  \beta'_{g_1}( g_{1\pm}) = \pm 2\sqrt{ (\beta_1^{g})^2 -\beta_0^{g}\beta_2^{g} }
     = \pm \sqrt{-g^2} \left( 4\frac{\Gamma(\frac{d}{4})^2}{ \Gamma(\frac{d}{2})}\right) + O(g^3)\;.
 \ee
 The beta function $\beta_{g_2}$ admits two fixed points and critical exponents of the same form, with $g\to\sqrt{3}g$.

\section{The Beta function coefficients}
\label{sec:divergences}
  
In this section we study the coefficients $\beta_0^{g}, \beta_1^{g}$ and $\beta_2^{g}$. 
As presented in Eq.~\eqref{eq:coef} they are ratios of sums of amplitudes of graphs which can be arbitrarily UV divergent. We will now show that:
\[
 - \beta_0^{g} = R \{ \pmb{\partial}U \} \;, \qquad 
 - \beta_1^{g} = R \{ \pmb{\partial}S \}  \;, \qquad
 - \beta_2^{g} = R \{ \pmb{\partial}T \} \;,
\]
with $R$ the BPHZ \cite{bogoliubow1957multiplikation,Hepp:1966eg,Zimmermann:1969jj} subtraction operator. 

\paragraph{The scale derivative.}
 
Let us consider an amputated graph $\cG$, and let $E(\cG)$, $n(\cG)$, $r(\cG)$ be the numbers of edges, vertices, and external half edges of $\cG$, respectively. Below, we are interested in 
chain graphs with $r(\cG)=4$. We have $ 2 E(\cG) =4 n(\cG) -r(\cG)$, and the amplitude of $\cG$ from Eq.~\eqref{eq:ampli1} is:
\be
 A(\cG) = \int_{\Lambda^{-2}}^{k^{-2}} 
  \left( \prod_{e\in \cG} d\alpha_e \, \alpha_e^{\zeta-1}\right) \; 
 \frac{ 1 }{ \big[ U_{\cG}(\alpha)\big]^{d/2} }  \;, \qquad U_{\cG}(\alpha) = \sum_{ {\cal T} \subset \cG} \prod_{e\notin {\cal T} } \alpha_e \;,
\ee
with $U_{\cG}(\alpha)$ the first Symanzik polynomial (see for instance \cite{Rivasseau:1991ub,Krajewski:2008fa}) of $\cG$. Clearly,  $U_{\cG}(\alpha)$ is a sum of positive terms and each monomial has the global degree $E(\cG)- n(\cG) +1$ in the variables $\alpha$.
We are interested in the scale derivative of the amplitude. After uniformly rescaling all the parameters by $k^{-2}$ and using $r(\cG)=4$, we have:
\be\label{eq:amplirescale}
  \pmb{\partial} A(\cG) =   -2 \sum_{e_0 \in \cG} 
  \int_{k^2\Lambda^{-2}}^{1} 
  \left(  \prod_{e\in \cG}^{e\neq e_0} d\alpha_e \, \alpha_e^{\zeta-1}\right)\; 
  \frac{ 1 }{ \big[ U_{\cG}(\alpha)
  \big]_{\alpha_{e_0}=1}^{d/2} } \; ,  
\ee
that is the scale derivative is a sum over graphs where one of the (rescaled) parameters has been set to its maximal value 1. 
We call this edge \emph{the marked edge} and we represent it as dashed.
Below we will need to consider the cap with zero rungs $\bullet$ which consists in a vertex (a blue vertex associated to $\lambda_1$) joined by two edges to the rest of the graph.

The derivatives can act on the horizontal or the vertical edges in a graph.
We call $H$ (for horizontal) the subgraph with two parallel horizontal edges, one of which is marked and $V$ (for vertical) the vertical rung with one marked edge. 
Correspondingly, we denote
$[ U_p H U_q ]$  the amplitude of a graph consisting in a ladder with $p\ge 1$ rungs followed by two horizontal edges, one of which is marked, followed by a ladder with $q\ge 1$ rungs
 and $[U_p V U_q]$ the amplitude of a graph consisting in a ladder with $p\ge 1$ rungs followed by a rung with a marked edge, followed by a ladder with $q\ge 1$ rungs. We denote their generating functions:
 \[ UHU = \sum_{p,q\ge 1}  g^{2p+2q} [ U_pHU_q ] \;, \qquad g^2 U VU  = \sum_{p,q\ge 1}  g^{2p+2q+2} [ U_p V U_q] \;. \] 
Several examples are depicted in Fig.~\ref{fig:Uh}. 
 
\begin{figure}[ht]
\begin{center}
\includegraphics[scale=1]{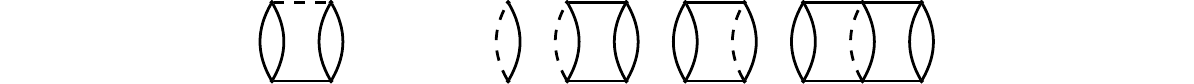} 
 \caption{From left to right, the graphs $ [ U_1 H U_1]$ and 
 $ V , [V U_1] ,  [U_1V ]$ and $[U_1VU_1]$.} \label{fig:Uh}
 \end{center}
\end{figure}
The scale derivative of $U_r$ is:
\begin{align*}
&\pmb{\partial} U_1 = (-4) V \;,\\ 
&\pmb{\partial} U_r = 
  (- 4) \bigg[  \sum_{p,q\ge 1}^{p+q = r}  [ U_p H U_q ]  + \sum_{p,q\ge 1}^{p+q = r-1}   [ U_p V U_q ] + [ V U_{r-1} ] + [ U_{r-1} V ] \bigg] \;,
   \;\;\;\; \forall r\ge 1\;,
\end{align*}
which for generating functions becomes:
\be\label{eq:derivU}
 \pmb{\partial} U = ( -4 ) \bigg\{ g^2  V   + 2 g^2 VU  + U ( H + g^2 V ) U  \bigg\}\;.
\ee
Similarly, we get, using obvious notation, (see Fig.~\ref{fig:Th}):
 \begin{align}\label{eq:derivT}
 & \pmb{\partial} T  = (-4 ) \bigg\{
  ( \bullet + S) (H +g^2 V ) ( \bullet + S ) \bigg\} \;,  \\ \label{eq:derivS}
& \pmb{\partial } S   =  (-4) \bigg\{
 (\bullet +S) \big[ g^2  V + ( H + g^2 V) U \big] \bigg\} \;.
 \end{align}

\begin{figure}[ht]
\begin{center}
\includegraphics[scale=1]{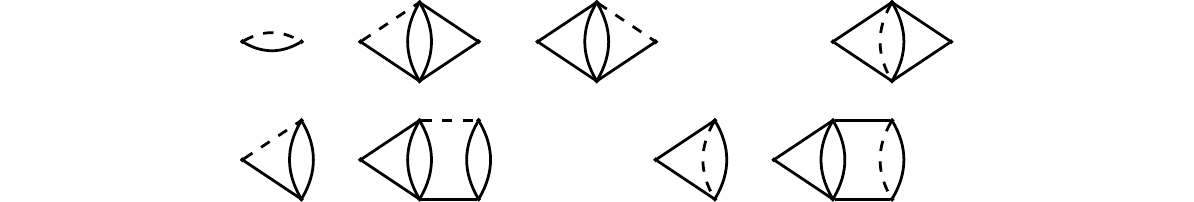} 
 \caption{Top row, from left to right the graphs $[\bullet H\bullet ] , 
 [\bullet H S_1], [S_1 H \bullet]  $ and $[\bullet  V \bullet]$.
 Bottom row, from left to right $[\bullet HU_1],[S_1HU_1]$ and
 $[\bullet V ],[S_1V ]$.
 } \label{fig:Th}
 \end{center}
\end{figure}

\paragraph{Taylor operators.}
A connected amputated subgraph $\gamma \subset \cG$ is a subset of the edges 
of $\gamma$. Like $\cG$, $\gamma$ contains all the end vertices of its edges. The external half edges (or legs) of $\gamma$ are either half edges of $\cG$ or come from the edges of $\cG$ which do not belong to $\gamma$ but are incident to vertices in $\gamma$. 

For any subgraph $\gamma$ of $\cG$, we denote $\tau_{\gamma}$ the ``localization'' operator acting on $A(\cG)$ which on the one hand separates the subgraph $\gamma$ and on the other, in $\cG$, it shrinks $\gamma$ to a vertex to give the graph $\cG/\gamma$:
\[
 \tau_{\gamma} A(\cG) = A(\cG/ \gamma) A(\gamma) \;. 
\]
The edges of $\cG$ are partitioned between the edges of $\gamma$ and the ones of $\cG/\gamma$. We denote $ \alpha_{\cG/\gamma} = \{ \alpha_e| e\in \cG/\gamma \} $ and $\alpha_{\gamma} = \{ \alpha_e| e\in \gamma \}$. Observe that 
$E(\cG) = E(\gamma) + E(\cG /\gamma)$ and $n(\cG) = n(\gamma) + n(\cG /\gamma) - 1$. Any spanning tree in $\cG$ is:
\begin{itemize}
 \item either the union of a spanning tree in $\gamma$ with a spanning tree in $\cG/\gamma$, in which case the global degree in  $\alpha_{\gamma}$ of the corresponding term in $U_{\cG}(\alpha)$ is exactly $ E(\gamma) - n(\gamma) +1$. Any tree in $\gamma$ and any tree in $\cG/\gamma$ lead to exactly one tree in $\cG$.
 \item or not, in which case the global degree in $\alpha_{\gamma}$ of the corresponding term in $U_{\cG}(\alpha)$ is at least $ E(\gamma) - n(\gamma) +2$.
\end{itemize}
It follows that under a uniform rescaling of $\alpha_{\gamma}$ by $u$ we have:
\[
 U_{\cG}(\alpha)\big{|}_{\alpha_{\gamma} = u\alpha_{\gamma}}  = 
 u^{   E(\gamma) - n(\gamma) + 1  }  
 \bigg{[} U_{\cG/\gamma}(\alpha_{\cG/\gamma} )U_{\gamma}(\alpha_{\gamma}) 
    + \sum_{q\ge 1} u^q 
    [ \alpha_{\gamma} ]^{E(\gamma) - n(\gamma) + 1 + q}  
    [\alpha_{\cG/\gamma}]^{ E(\cG/\gamma) - n(\cG/\gamma) +1 -q }
 \bigg{]} \;,
\]
where we indicated the global scaling with $\alpha_{\gamma}$
and $\alpha_{\cG/\gamma} $ of the corrections.
The localization operator $\tau_{\gamma}$ can be implemented as a Taylor operator
\cite{Bergere:1974zh,Bergere:1977ft,Bergere:1980sm}  acting on the integrand:
\begin{align}\label{eq:taylor}
 \tau_{\gamma}\frac{ 1 }{ [U_{\cG}(\alpha)]^{d/2}}  & = 
                            \frac{ u^{\frac{d}{2} \left[ E(\gamma) - n(\gamma) + 1\right] }  }{ \big[U_{\cG}(\alpha) \big{|}_{\alpha_{\gamma} = u\alpha_{\gamma}} \big]^{d/2}} \bigg|_{u\to 0} 
       = \frac{1}{[U_{\cG/\gamma}(\alpha_{\cG/\gamma})]^{d/2}} \;
       \frac{1}{[U_{\gamma}(\alpha_{\gamma})]^{d/2}}    \;, \crcr
(1 - \tau_{\gamma} )\frac{ 1 }{ [U_{\cG}(\alpha)]^{d/2}} & = 
 \int_0^1 du\; \frac{d}{du} \bigg\{     u^{\frac{d}{2} \left[ E(\gamma) - n(\gamma) + 1\right] }   \frac{ 1 }{ \big[U_{\cG}(\alpha) \big{|}_{\alpha_{\gamma} = u\alpha_{\gamma}} \big]^{d/2}}      \bigg\} \;.
\end{align}

\paragraph{Subtraction operator.}
We call the one particle irreducible four-point subgraphs\footnote{As we deal with graphs with no two-point subgraphs, all the connected four-point subgraphs are automatically one particle irreducible.} of 
$\cG$  \emph{dangerous}. 
An inclusion forest of dangerous subgraphs is a set of dangerous subgraphs which are either \emph{nested} or \emph{totally disjoint} (that is they do not have any vertex in common):
\[
 \bigg\{ \gamma \subset \cG \;, r(\gamma) =4 \; \bigg{|} \;\; 
  \forall \gamma_1,\gamma_2 \; \text{ either } \gamma_1 \subset \gamma_2\, \text{ or } \gamma_2 \subset \gamma_1 \, \text{ or }
   \gamma_1\cap \gamma_2 =\emptyset \bigg\} \;.
\]
We denote $\pmb{F}(\cG)$ the set of all the inclusion forests of dangerous subgraphs of $\cG$, including the empty forest. The BPHZ subtraction operator \cite{bogoliubow1957multiplikation,Hepp:1966eg,Zimmermann:1969jj}
is:
\[
 R  =\sum_{F \in \pmb{F}(\cG)} \prod_{\gamma \in F} (-\tau_{\gamma}) \;.
\]
The operator is well defined because the localization operators of graphs in a forest commute. 

We are concerned here with ladders, caps and double caps $U,S$ and $T$. In all these cases the dangerous proper subgraphs (i.e. different from the graph itself) have a  particularly simple structure.

\begin{lemma}
Any proper four-point subgraph $\gamma \subset \cG$ consists in a sequence of vertical rungs connected by horizontal edges. $\gamma$ can reach one end of the graph or not (the end is either a vertex with two external points for $T$ and for one end of $S$, or a rung with two external vertices for $U$ and for the other end of $S$).
\end{lemma}

\begin{proof}
 Consider $\cG$ a ladder, cap or double cap and a four-point proper subgraph $\gamma\subset \cG$ (that is $\gamma$ is not $\cG$ itself):
\begin{itemize}
\item assume $\gamma$ does not contain any internal vertex of $\cG$ (i.e. a vertex which is not incident to external half edges of $\cG$). This is impossible for $S$. For $U$, $\gamma$ must consist in 
a vertical rung connecting two of the boundary vertices. For $T$, 
only $T_0$ has such a subgraph, $\gamma = T_0$, which is not proper.

\item assume $\gamma$ contains an internal vertex $v$. Then this vertex is part of a rung with edges $e_1,e_2$ connecting $v$ with $v'$.
    \begin{itemize}
     \item assume that neither $e_1$ nor $e_2$ belong to $\gamma$. Then both the horizontal edges incident at $v$ must belong to $\gamma$ (otherwise $\gamma$ is 1PR), and $v'$ must also belong to $\gamma$ (otherwise again $\gamma$ is 1PR). But $v$ and $v'$ already support four external points for $\gamma$, hence $\gamma$ has no additional external points, which is impossible.
     \item assume only $e_1$ belongs to $\gamma$ but $e_2$ does not.
     Consider the four horizontal edges incident to $v$ and $v'$
        \begin{itemize}
         \item  if neither one of them or only one of them belongs to $\gamma$ then $\gamma$ has more than four external points.
         \item  if exactly two of them belong to $\gamma$ then 
         there are already four external points supported by $v$ and $v'$, hence $\gamma$ can not have any additional external points which is impossible. 
         \item if only three of them belong to $\gamma$ then $\gamma$ is 1PR. 
         \item if all the four belong to $\gamma$, then $\gamma$ splits into a left part (to the left of the rung $e_1,e_2$) and a  right part (to the right of $e_1,e_2$). Each part brings at least two additional external points, which makes $\gamma$ at least a six point graph.
        \end{itemize}
    \end{itemize}
\end{itemize}

It follows that $\gamma$ must contain a rung with edges $e_1e_2$ connecting two internal vertices $v$ and $v'$. The iteration is now simple 
\begin{itemize}
 \item either $\gamma$ consists in only this rung,
 \item or  (as $\gamma$ is 1PI) it contains the pair of horizontal edges incident to $v$ and $v'$ pointing to the right (or the pair pointing to the left or both). These edges either
     \begin{itemize}
      \item reach the end of the graph which is either an external vertex with two external points ($T$ or one end of $S$) or a rung with two external points ($U$ or the other end of $S$). In the second case the rung must belong to $\gamma$, as two external points must come from the left of $\gamma$.
      \item reach a pair of internal vertices $w$ and $w'$ connected by a rung $e_1',e_2'$. But then the entire rung $e_1',e_2'$ belongs to $\gamma$ and we iterate.
     \end{itemize}
\end{itemize}

\end{proof}

Now, the graphs contributing to $\pmb{\partial}T,~\pmb{\partial}S$ and 
$\pmb{\partial}U$ have the additional marked edge whose parameter is set to $1$. In this case, we restrict to dangerous subgraphs which do not contain the marked edge. Then the dangerous subgraphs are 
confined to live either to the left or to the right of the marked edge,
and the subtraction operator factors into the sum over the forests of left subgraphs and the sum over forests of right subgraphs:
\[
  R = R^{\rm left} R^{\rm right } \;,\qquad 
 R^{\rm left}  =\sum_{F^{\rm l} \in \pmb{F}^{\rm l}(\cG)} \prod_{\gamma^{\rm l} \in F^{\rm l}} (-\tau_{\gamma^{\rm l}}) \; ,\;\;
 R^{\rm right}=    \sum_{F^{\rm r} \in \pmb{F}^{\rm r}(\cG)} 
   \prod_{\gamma^{\rm r} \in F^{\rm r}} (-\tau_{\gamma^{\rm r}}) \;,
\]
where the forests in $\pmb{F}^{\rm l}(\cG)$ (respectively 
$\pmb{F}^{\rm r}(\cG)$) contain only graphs at the left (resp. right) 
of the marked edge. Let us discuss the forest of left subgraphs. We denote the rungs to the left of the marked edge by $1,2\dots p$ ($p$ being the closest to the marked edge). We treat the case in which the left end of the graph is a cap (a similar reasoning works for the case of an end rung). We denote $S_r$ the subgraph starting at the left end of $\cG$ and having $r$ rungs. Any left forest can be obtained from a forest which does not contain $S_p$ by adding or not $S_p$. Thus:
\[
 R^{\rm left} = (1-\tau_{S_p} )
 \sum_{F^{\rm l} \in \pmb{F}^{\rm l}(\cG)}^{S_p \notin F^{\rm l}} \prod_{\gamma^{\rm l} \in F^{\rm l}} (-\tau_{\gamma^{\rm l}}) \; ,
\]
Now, among the graphs $\gamma^{\rm l}$ in the forest $F^{\rm l}$ some, denoted
$\gamma^{\rm l}_{\supset p}$, contain the rung $p$, and the rest do not. We have
\[
 (1-\tau_{S_p} ) \tau_{ \gamma^{\rm l}_{\supset p} } =0 \;,
\]
therefore the sum truncates to the forests such that none of the graphs in the forest contains the last rung. Iterating we find:
\[
 R^{\rm left} = \prod_{i=1}^p (1-\tau_{S_i} ) \;.
\]
Let us define
\[
 R^S = \prod_{i\ge 1} (1-\tau_{S_i} ) \;, \qquad
  R^U = \prod_{j\ge 1} (1-\tau_{U_j} ) \;, 
\]
where $\tau_{S_p}A(\cG)$ (resp. $\tau_{U_p}A(\cG)$) is  zero if $S_p$ (resp.  $U_p$) is not a subgraph of $\cG$.
We can now state the main result of this section.
\begin{theorem}\label{thm:renorm}
 The coefficient $\beta^{g}_0$, $\beta^{g}_1 $ and $\beta^{g}_2$ are minus the renormalized scale derivatives of the ladders, caps and double caps generating functions:
 \begin{align*}
-  \beta_0^{g} &=  \pmb{\partial} U - 2 \frac{U}{1+S} \pmb{\partial} S + \frac{U^2}{(1+S)^2} \pmb{\partial} T   =R \big\{ \pmb{\partial} U  \big\} =
   \big( R^U\big)^{\rm left} \big(R^U\big)^{\rm right} \big\{ \pmb{\partial} U  \big\}
 \; , \crcr
-  \beta_1^{g} &= \frac{1}{1+S} \pmb{\partial}S - \frac{U}{(1+S)^2} \pmb{\partial} T    = R \big\{ \pmb{\partial} S  \big\}  =    \big( R^S\big)^{\rm left} \big(R^U\big)^{\rm right}  \big\{\pmb{\partial} S\big\}
 \;,  
 \crcr 
 -  \beta_2^{g} &=  \frac{1}{(1+S)^2} \pmb{\partial} T =R \big\{ \pmb{\partial} T \big\}
 =   \big( R^S\big)^{\rm left} \big(R^S\big)^{\rm right}  \big\{ \pmb{\partial} T \big\}
 \; .
 \end{align*}
\end{theorem}

\begin{proof} Recalling the scale derivatives from 
equations \eqref{eq:derivU}, \eqref{eq:derivT}, and \eqref{eq:derivS}, 
the theorem follows provided that, for $\cG = H,V$, we have:
\[
 R^S \big\{ (  \bullet +S ) \cG  \big\}  = 
 \left( \frac{1}{1+S}  \right)
    \big[ ( \bullet +S ) \cG   \big]
 \;,\qquad  R^U \big\{  \cG  U  \big\} = 
    \cG U  -  
    \left( \frac{U}{1+S} \right) \big[ \cG ( \bullet +S )  \big] \;.
\]
Let us first check $R^S$. We have:
\be
\begin{split}
  R^S \big\{ ( \bullet +S ) \cG  \big\}  & =  \bullet \cG  +  \sum_{r\ge 1} g^{2r} \left[ \prod_{i=1}^r
  (1-\tau_{S_i}) \right] [ S_r\cG ] \crcr
 & =  \bullet \cG   +  S  \cG  + 
  \sum_{r\ge 1} g^{2r}  \sum_{q=1}^r 
   \sum_{1\le i_1<\dots <i_q \le r}
 \left[ \prod_{s=1}^q   (-\tau_{S_{i_s}} )  \right] [S_r\cG] \;.
 \end{split}
\ee
Taking into account the action of the localization operators on the cap, the sum over $r$ becomes:
\begin{align*}
&  \sum_{r\ge 1} g^{2r}  \sum_{q=1}^r (-1)^q
   \sum_{1\le i_1<\dots <i_q \le r}
  S_{i_1} S_{i_2 - i_1} \dots S_{i_{q}-i_{q-1}}
 [ S_{r-i_q }\cG] \crcr
 & = \sum_{r\ge 1} g^{2r}  \sum_{q=1}^r (-1)^q
   \sum_{d_1, \dots d_{q} \ge 1}^{d_1 + \dots + d_{q} \le r}
  S_{d_1}  \dots S_{d_{q}} [ S_{r- d_1 \dots - d_q  } \cG ] \crcr
  & =  
    \sum_{q\ge 1} (-1)^q
  \left( \sum_{d\ge 1} g^{2d} S_{d} \right)^q  
  \left(   \bullet \cG   +\sum_{p\ge 1} g^{2p} [S_p \cG]  \right)
   = \left( \frac{1}{1+S} -1 \right)
   \big[  ( \bullet + S) \cG \big] \;.
\end{align*}
Concerning $R^U$, we have:
 \begin{align*}
  R^U \big\{  \cG U   \big\} &=     \cG U  
  + \sum_{r\ge 1} g^{2r} \sum_{q=1}^r (-1)^q \sum_{1\le i_1<\dots <i_q \le r}
 \left[ \prod_{s=1}^q   ( \tau_{U_{i_s}} )  \right] [ \cG U_r]
 \crcr
 & = \cG U 
  + \sum_{r\ge 1} g^{2r} \sum_{q=1}^r (-1)^q \sum_{1\le i_1<\dots <i_q \le r}
  U_{i_1} S_{i_2-i_1} \dots S_{i_q - i_{q-1}} [ \cG S_{r-i_q} ] \crcr
  & = \cG U + \sum_{q\ge 1} (-1)^q 
  \left(  \sum_{d\ge 1} g^{2d} U_d \right)
   \left( \sum_{d\ge 1} g^{2d}S_d \right)^{q-1} \left( 
    \cG\bullet  + \sum_{p\ge 1} g^{2p} [ \cG S_p ]  \right) \crcr
   & = \cG U - \frac{U}{1+S} \big[  \cG ( \bullet +S) \big] \;.
 \end{align*}

\end{proof}

At first orders the bare amplitudes are: 
\begin{align*}
\left( - \frac{1}{4} \right)  \pmb{\partial} U =&  
 g^2 V  + 2g^4 VU_1 + g^4 U_1 H U_1 + O(g^6) \;,\crcr
 \left( - \frac{1}{4} \right)  \pmb{\partial} S = & 
 g^2 \bullet V + g^2 \bullet H U_1 \;, 
 \qquad \left( - \frac{1}{4} \right)  \pmb{\partial} T = \bullet H \bullet 
 + g^2 \bullet V \bullet +2 g^2  \bullet H S_1 \; ,
\end{align*}
and the first non trivial renormalized amplitudes are (see Fig.~\ref{fig:betaeq}):
\begin{align*}
    \frac{ (- 1) }{4}  R[\pmb{\partial}U] \bigg{|}_{g^4} & = 
    2(1-\tau_{U_1}) [VU_1] + (1-\tau_{U_1}) (1-\tau_{U_1})  [ U_1 H U_1 ] \;,
 \crcr
  \frac{ (- 1) }{4}  R[\pmb{\partial}S] \bigg{|}_{g^2}& =
  \bullet V + (1-\tau_{U_1} ) \bullet [HU_1] \;, \qquad
  \frac{ (- 1) }{4}  R[\pmb{\partial}T] \bigg{|}_{g^2}  = 
  \bullet V \bullet +2  (1 - \tau_{S_1}) \bullet H S_1 \;.
\end{align*}
This should be compared to the coefficients in Eq.~\eqref{eq:beta3loops}.

\begin{figure}[ht]
\begin{center}
\includegraphics[scale=1]{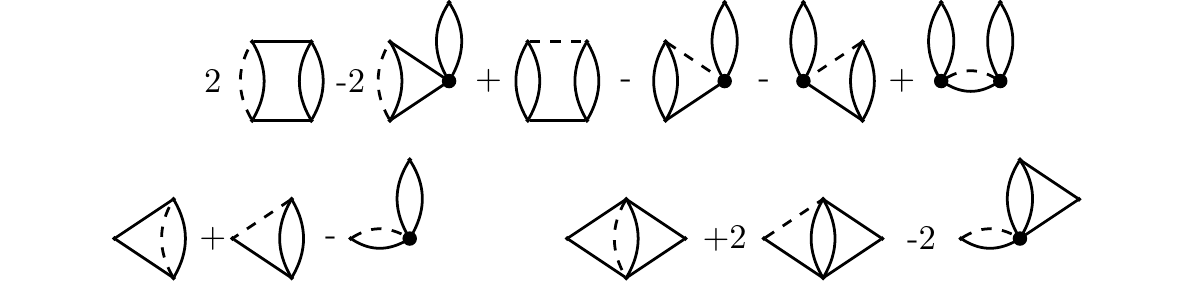} 
 \caption{First non trivial subtractions. We represented the shrunk vertex as thickened.} \label{fig:betaeq}
 \end{center}
\end{figure}

\paragraph{Convergence.} It remains to prove that the subtracted amplitudes are convergent. This might seem obvious at first sight: this is the whole 
``raison d'\^etre'' of the subtraction operator. However, it turns out that for massless theories subtracted a zero momentum the usual proofs do not work. What works is the
proof in \cite{Rivasseau:1991ub} chapter II.3 relying on multiscale analysis and the classification of forests.
The analysis is however quite involved and we will not reproduce it here.

\section{Spectrum of bilinear operators}
\label{sec:spectrum}

We wish to further compare our results to those of \cite{Klebanov:2016xxf,Giombi:2017dtl}. We will do that by slavishly following their computation of the spectrum of invariant bilinear operators, adapted to our model. This computation will be less rigorous than those in the previous sections, but it will be nonetheless instructive. We leave a more rigorous renormalization group derivation of such dimensions for future work.

First of all, we assume that we are at a fixed point, hence we use conformal field theory methods. Second, we will only consider the tetrahedron interaction.
Let us write the three point function of a spin-zero primary operator $\mathcal{O}_h$ (assumed to be the continuation of a bilinear operator $\varphi_{abc} (\partial_{\mu}\partial^{\mu})^{n} \varphi_{abc}$, for some integer $n$, from the free theory to the interacting one) with two fields $\phi^{abc}$ in the conformal theory as:
\begin{equation}
v(x_0,x_1,x_2)=\langle \mathcal{O}_h(x_0)\varphi_{abc}(x_1)\varphi_{abc}(x_2)\rangle=\frac{C_{\mathcal{O}\phi\phi}}{(x_{01}^2)^{h/2}(x_{02}^2)^{h/2}(x_{12}^2)^{d/4-h/2}} \;.
\end{equation}
The three point function in the conformal theory satisfies the Bethe-Salpeter equation:
\begin{equation}
v(x_0,x_1,x_2)=\int d^dx_3d^dx_4 \; K(x_1,x_2,x_3,x_4)v(x_0,x_3,x_4) \;.
\label{eq:BSE}
\end{equation}
Having neglected the pillow and double-trace terms, the kernel $K$ is just the tetrahedron kernel:\footnote{The pillow and double-trace terms contain a $\delta(x_{34})$, which, when used in \eqref{eq:BSE}, leads to zero contribution for $h>d/2$, and a divergence for $h<d/2$. The latter reads $\int_x \delta(x) |x|^{h-d/2} = \int_p \int_x e^{\im p x} |x|^{h-d/2} = \int_p |p|^{-h-d/2}$, which is zero in dimensional regularization.
}
\begin{equation}
\label{kernel}
K(x_1,x_2,x_3,x_4)=3\lambda^2G(x_{13})G(x_{24})G(x_{34})^2 \;,
\end{equation}
and the two-point function in position space is:
\begin{equation}
G(x)=\frac{1}{(2\pi)^{d/2}Zx^{d/2}}\;.
\end{equation}

Using the same integral formulas as in \cite{Klebanov:2016xxf,Giombi:2017dtl}, we find:
\begin{equation}
\int d^dx_3d^dx_4 \; K(x_1,x_2,x_3,x_4)v(x_0,x_3,x_4)=f(h)v(x_0,x_3,x_4) \;,
\end{equation}
with
\begin{align}
\label{spectrum}
f(h)=\frac{3\lambda^2}{(4\pi)^d Z^4}L_d(\frac{d}{4},\frac{h}{2})L_d(\frac{d-h}{2},\frac{d}{4}) 
=\frac{3g^2}{(4\pi)^d}\frac{\Gamma(h/2-d/4)\Gamma(d/4-h/2)}{\Gamma(3d/4-h/2)\Gamma(d/4+h/2)} \;,
\end{align}
where we have used $g=Z^{-2}\lambda$.
We can use the rescaled coupling $\tilde{g}=g(4\pi)^{-d/2}\Gamma(d/4)^{-2}$, and forgetting the tilde, we obtain:
\begin{equation}
f(h)=3g^2\Gamma(d/4)^4\frac{\Gamma(h/2-d/4)\Gamma(d/4-h/2)}{\Gamma(3d/4-h/2)\Gamma(d/4+h/2)} \;.
\end{equation} 
The dimension $h$ of the spin zero operators are determined by $f(h)=1$. 

We will now analyze the solutions of this equation.  The main outcome of this analysis is that for every $d$, if the tetrahedron coupling is purely imaginary and not too large, then the dimensions of all the operators are real.

\paragraph{\it Solutions for $d=3$.}

For $d=3$, $f(h)$ becomes:
\begin{equation}
f(h)=3g^2\Gamma(3/4)^4\frac{\tan (\pi(1/4-h/2))}{(5/4-h/2)(1/4-h/2)(h/2-3/4)} \;.
\end{equation}
Up to second order in $g$, the solutions are:
\begin{align} 
h_{0\pm}&=\frac{3}{2}\pm 4 \sqrt{-\frac{3g^2}{\pi} }\, \Gamma(3/4)^2  +\mathcal{O}(g^3) \;,\\
h_n&=\frac{3}{2}+2n+\frac{24\Gamma(3/4)^4g^2}{\pi n(2n+1)(2n-1)} +\mathcal{O}(g^3) \;, \;\;\; n \in \mathbb{N}^{+} \;.
\end{align}

\paragraph{\it Solutions for $d=2$.}
 
For $d=2$, $f(h)$ becomes:
\begin{align}
f(h)=-3g^2\pi^2\frac{4}{(1-h)^2}\; , 
\end{align}
with solutions
\be \label{eq:2d-hsol}
h_{\pm}=1\pm 2\pi \sqrt{-3 g^2}   \;.
\ee
As in \cite{Giombi:2017dtl} we find a degeneration of the solutions in two dimensions.
It can be checked numerically that in $d=2+\epsilon$ there are additional solutions that approach $h_n= 1+2n$ for $\epsilon\to 0$, but they are missed if we take first the limit $\epsilon\to 0$ of $f(h)$.
This can be understood from the following observation: for $d>2$ and  $ n \in \mathbb{N}^{+}$, we have $\lim_{\delta\to 0^{\pm}} f(d/2+2n+\delta) = \mp \infty$, while for $d<2$ the limit is reversed, $\lim_{\delta\to 0^{\pm}} f(d/2+2n+\delta) = \pm \infty$, as exemplified in Fig.~\ref{fig:BSeq-2d}. Therefore, in the limit $d\to 2^{\pm}$, the solutions of $f(h)=1$ approach $h_n=1+2n$.

\begin{figure}[ht]
\centering
\begin{minipage}{0.4\textwidth}
           \centering 
            \includegraphics[width=\textwidth]{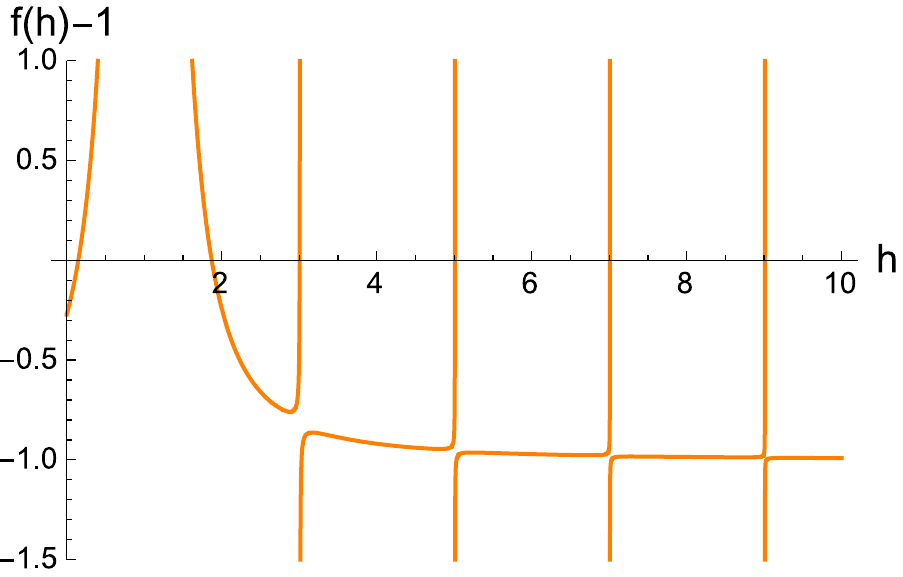}
        \end{minipage}
        \hspace{0.01\textwidth}
\begin{minipage}{0.4\textwidth}
            \centering
            \includegraphics[width=\textwidth]{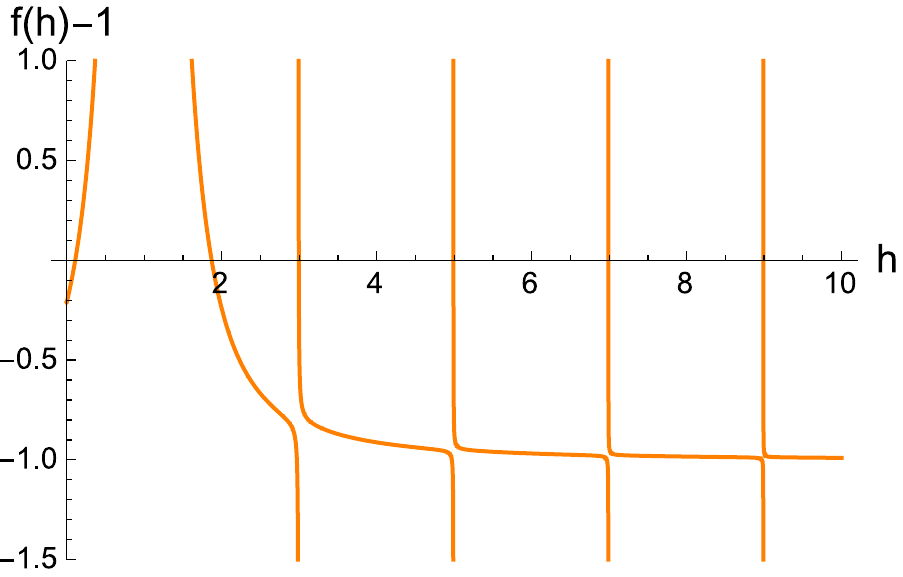}
        \end{minipage}
 \caption{Plots of $f(h)-1$ at $g=0.08$, for $d=2.02$ (left), 
 and  $d=1.98$ (right).} \label{fig:BSeq-2d}
\end{figure}

\paragraph{\it Solutions for $d=1$.}

For $d=1$, $f(h)$ becomes:
\begin{align}
f(h)=3g^2\Gamma(1/4)^4\frac{\tan(\pi(h/2+1/4))}{h/2-1/4} \;.
\end{align}
Up to second order in $g$ the solutions are:
\begin{align}
h_{0\pm}&=\frac{1}{2}\pm 2 \sqrt{-\frac{3g^2}{\pi}}\,\Gamma(1/4)^2 +\mathcal{O}(g^3) \;,\\
h_n&=\frac{1}{2}+2n-\frac{6\Gamma(1/4)^4g^2}{n\pi} +\mathcal{O}(g^3) \;, \;\;\; n\in \mathbb{N}^{+} \;.
\end{align}

\paragraph{\it Solutions for any $d$.}
In any dimension we find two solution $h_{0\pm}$ which are complex for real $g$, and an infinite sequence of real solutions $h_n$. For $g\to 0$, the latter approach $d/2+2n$, the classical scaling dimension of the bilinear operators $\varphi_{abc} (\partial_{\mu}\partial^{\mu})^{n} \varphi_{abc}$. As $n\geq 1$, they are consistent with the assumption $h>d/2$.

The solutions $h_{0\pm}$ approach for $g\to 0$ the classical scaling dimension of $\varphi_{abc}\varphi_{abc}$, while for $g\neq 0$ they are of the form $h_{0\pm}=\frac{d}{2}\pm \im\alpha(d,g)$ as in the model of Klebanov and Tarnopolsky.
In fact, with $h=\frac{d}{2}\pm \im\alpha$, $f(h)$ becomes:
\begin{align}
f(h)=3g^2\Gamma(d/4)^4\frac{\Gamma(\im\frac{\alpha }{2})\Gamma(-\im \frac{\alpha }{2})}{\Gamma(d/2+\im\frac{\alpha }{2})\Gamma(d/2-\im \frac{\alpha }{2})} \; ,
\end{align}
which is solved, up to second order in $g$, by
$
\alpha(d,g)=\pm 2\sqrt{3}\frac{\Gamma(d/4)^2}{\Gamma(d/2)} g
$.
Therefore, for any $d$, we have a solution of the form 
\be\label{eq:h0}
h_{0\pm}=\frac{d}{2}\pm \im 2\sqrt{3}\frac{\Gamma(d/4)^2}{\Gamma(d/2)} g +\mathcal{O}(g^3)\;,
\ee
reproducing the values we found for integer $d$. 

It is these solutions which are complex for real tetrahedral coupling and lead to the instability discussed in \cite{Giombi:2017dtl}.
For $g$ real, both solutions $h_{0\pm}$ have  ${\rm Re}(h)=d/2$.
On the other hand, for purely imaginary $g$, they become real, and we have $h_{0+}>d/2$ and $h_{0-}<d/2$.  This is exactly what one should expect for an IR and UV fixed point, respectively. In fact $h_{0\pm}$ is the dimension of $\varphi_{abc}\varphi_{abc}$ at the fixed point. The latter can be obtained in the large-$N$ limit as half the dimension of the double trace invariant $(\varphi_{abc}\varphi_{abc})^2$, which is $(d+\nu)$, with $\nu$ the critical exponent of the double-trace coupling. This is given by 
Eq.~\eqref{eq:exp} with $g\to\sqrt{3}g$, and we obtain for the dimension of $\varphi_{abc}\varphi_{abc}$ precisely equation \eqref{eq:h0}. The IR fixed point is $g_{2+}$, corresponding  to $h_{0+}$, and the UV fixed point is $g_{2-}$, corresponding  to $h_{0-}$. Notice that the conformal dimensions at the IR and UV fixed points are related by $h_{0+}=d-h_{0-}$, as expected on general grounds for double-trace deformations of conformal field theories \cite{Gubser:2002vv}.

\paragraph{Special value of $g$.}

From now on we discuss the case $g$ (and $\lambda$) purely imaginary.
We denote $g_0 = 3^{-1/2}g_c (4\pi)^{-d/2}\Gamma(d/4)^{-2}$, with $g_c$ defined in Eq.~\eqref{eq:gc}. This $g_0$ is the maximal value of $|g|$ for which $\lambda(g)$ is invertible to $g(\lambda)$ as depicted in Fig.~\ref{fig:tetra} (we have also taken into account the rescaling implemented in this section). 

For $g=\im g_0$, we have two exact solutions at $h=0$ and $h=d$, for any $d$.\footnote{The solution $h=0$ violates the unitarity bound $h>d/2-1$ for $d>2$. Therefore, for $d>2$ we expect to find an upper value of $|g|$ beyond which the UV CFT with bilinear operator of dimension $h_{0-}$ is necessarily non-unitary.} From a numerical check, we find that for $d\lesssim 2.9728$ these correspond to $h_{0\pm}$, while for  $d\gtrsim 2.9728$ $h=0$ comes from the continuation of a negative solution, and $h=d$ corresponds to $h_1$.
This is depicted in Fig.~\ref{fig:BSeq-g0}.

Interestingly, we have a neat interpretation from the point of view of AdS/CFT. According to the standard dictionary \cite{Gubser:1998bc,Witten:1998qj} we should have $h_\pm=\frac{d}{2} \pm \sqrt{\frac{d^2}{4}+m^2}$, with $m$ being the mass of a field in AdS${}_{d+1}$. The plus and minus signs correspond to different boundary conditions \cite{Klebanov:1999tb}, the plus being always allowed, and the minus only for $-\frac{d^2}{4}<m^2<-\frac{d^2}{4}+1$. For $d\lesssim 2.9728$ we have 
the following situation. At $g=0$ we get $h_{0\pm}=d/2$ which can be interpreted
as $h_{\pm}$ with the mass saturating the Breitenlohner-Freedman bound $m^2\geq -\frac{d^2}{4}$ \cite{Breitenlohner:1982jf}. We dial up $g$ and when we reach $g=\im g_0$ the mass become zero; in this case $h_-$ is only allowed in the bulk for $d\leq 2$.
For $d\gtrsim 2.9728$ it is $h_1$ (instead of $h_{+}$) which starts at $g=0$ from a positive mass $m^2=4-\frac{d^2}{4}$ and reaches $m^2=0$ at $g=\im g_0$.

\begin{figure}[ht]
\centering
\begin{minipage}{0.4\textwidth}
           \centering 
            \includegraphics[width=\textwidth]{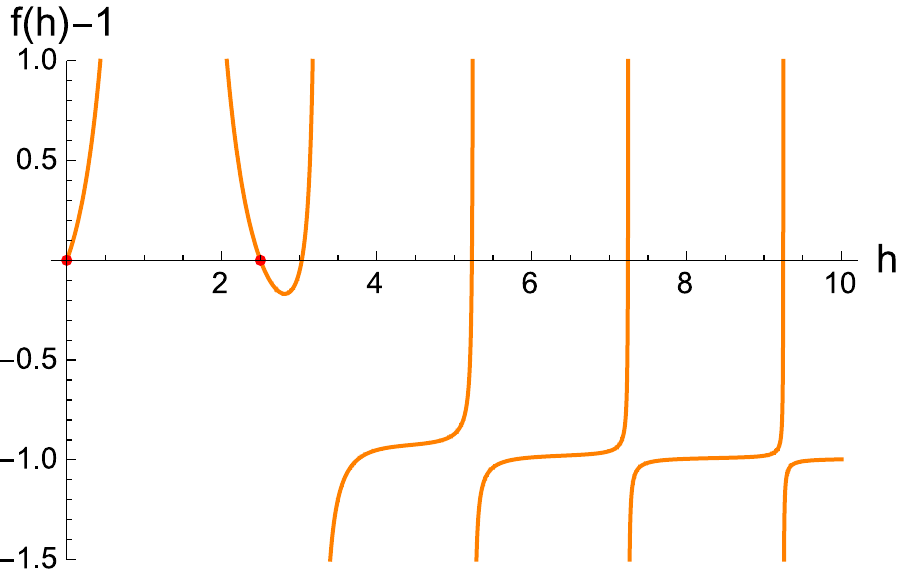}
        \end{minipage}
        \hspace{0.01\textwidth}
\begin{minipage}{0.4\textwidth}
            \centering
            \includegraphics[width=\textwidth]{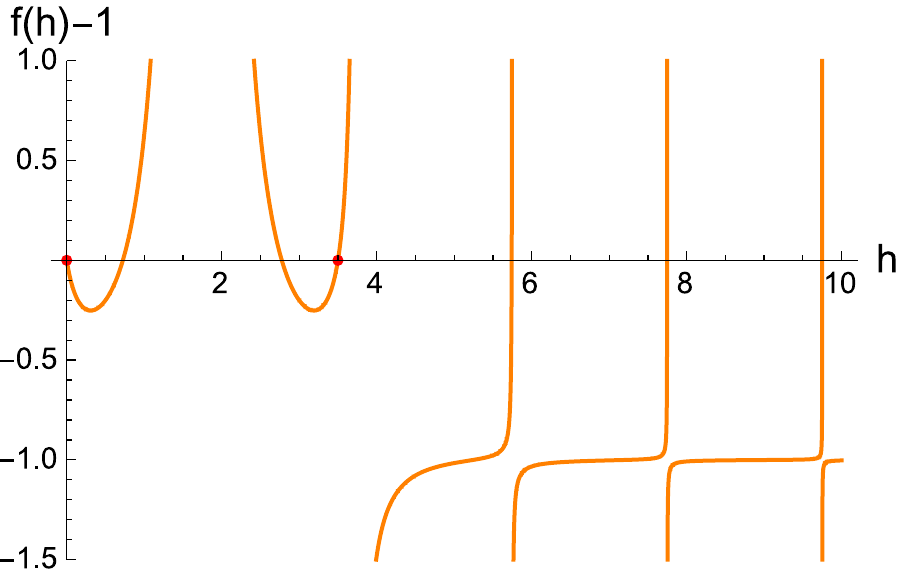}
        \end{minipage}
 \caption{Plots of $f(h)-1$ for $d=2.5$ (left) and $d=3.5$ (right) at $g=g_0(d)$. The zeros on the left panel correspond (from left to right) to $h_{0-}$, $h_{0+}$, and $h_1$ to $h_4$, while on the right panel this applies only to $h>0$, as $h=0$ is a new solution. The zeros at $h=0$ and $h=d$ are marked by a red dot.} \label{fig:BSeq-g0}
\end{figure}

\paragraph{Reappearance of complex dimensions.}

As already stated, for $g$ purely imaginary and small enough, we obtain a real spectrum. However, it is plausible that there exists a value $g_*$ such that for $\vert g \vert >  g_* $, some dimensions become complex again. 
With a numerical  study, we find that for $d>2$ there exists a $g_*\geq g_0$ at which $h_{0+}$ and $h_1$ merge and beyond which they become complex. The transition can be understood from the graphical solution of the equation, as in Fig.~\ref{fig:BSeq-3d}. 
Only at $d \simeq 2.9728$ we find that $g_*= g_0$, while $g_*\to+\infty$  for $d\leq 2$. The disappearance of the transition for $d\leq 2$ can be deduced from Fig.~\ref{fig:BSeq-2d} (notice also that at $d=2$, $h_{0-}$ becomes negative for $|g|>(2\pi\sqrt{3})^{-1}$ as deduced from Eq.~\eqref{eq:2d-hsol}).

It has been conjectured in \cite{Kim:2019upg} that the appearance of a complex scaling dimension $h$ with ${\rm Re}(h)=d/2$ is associated to a non-zero vacuum expectation value of the associated operator, hence to a spontaneous breaking of conformal invariance. The intuitive reason is that in the AdS/CFT picture such operators correspond to fields with mass below the Breitenlohner-Freedman bound. In our case, such a phenomenon seems to take place at the transition from $g^2<0$ to $g^2>0$. It is plausible that the instability in the AdS side of the correspondence translates into a spontaneous breaking of conformal invariance of our model at real coupling.
On the other hand, the complex dimensions at $g^2<-g_*^2<0$ have ${\rm Re}(h)>d/2$, so they seem to correspond to fields with complex mass (with both real and imaginary parts being non-zero). 
However, we stress that from a renormalization group point of view the model makes sense only for $|g|\leq g_0$. Since $g_*\geq g_0$ for any $d$, the appearance of such complex solutions is probably not relevant to our model.

\begin{figure}[ht]
\centering
\begin{minipage}{0.4\textwidth}
           \centering 
            \includegraphics[width=\textwidth]{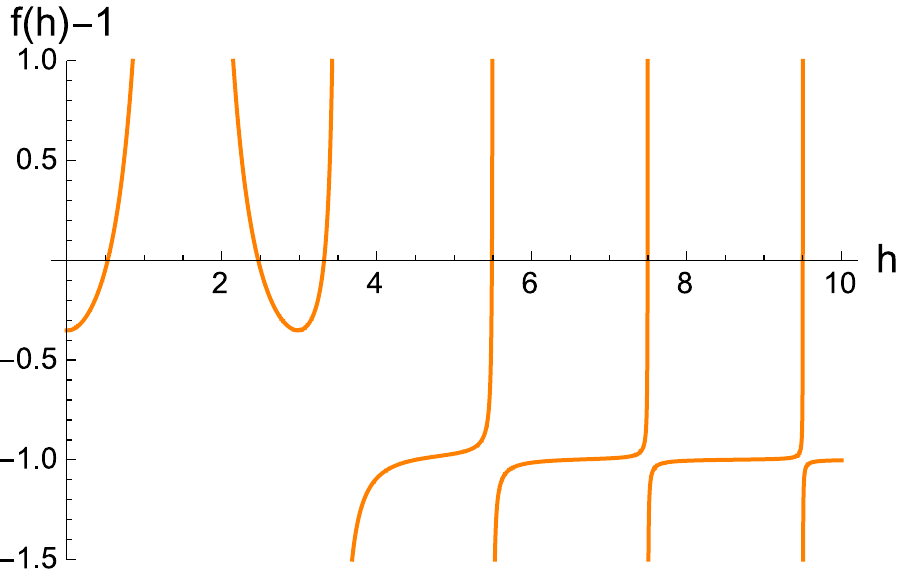}
        \end{minipage}
        \hspace{0.01\textwidth}
\begin{minipage}{0.4\textwidth}
            \centering
            \includegraphics[width=\textwidth]{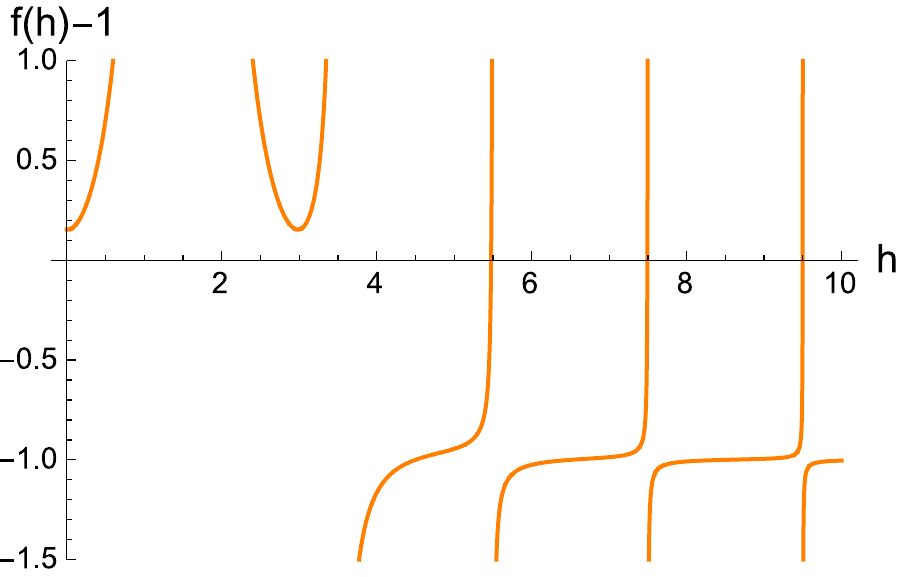}
        \end{minipage}
 \caption{Plots of $f(h)-1$ for $d=3$ at $g=0.15$ (left) and  $g=0.19$ (right). The zeros on the left panel correspond (from left to right) to $h_{0-}$, $h_{0+}$, and $h_1$ to $h_4$. On the right panel $h_2$ to $h_4$ are the only remaining real roots. In this case $g_0\simeq 0.186135$ and $g_*-g_0\simeq 5.7\times 10^{-5}$.} \label{fig:BSeq-3d}
\end{figure}


\section*{Acknowledgements}

\noindent 
We would like to thank Vincent Rivasseau and Igor Klebanov for helpful discussions.
\\
This research was supported in part by Perimeter Institute for Theoretical Physics. Research at Perimeter Institute is supported by the Government of Canada 
through the Department of Innovation, Science and Economic Development Canada and by the Province of Ontario through the Ministry of Research, Innovation and Science.

\newpage

 \appendix

\section{The melon integral}
\label{app:melon}

We show that:
\be
\begin{split}
M^{\Lambda}_k(p) & = \frac{1}{Z^3(4\pi)^d \Gamma(\zeta)^3 }    \int_{\Lambda^{-2}}^{k^{-2}}  d\alpha_1 d\alpha_2 d\alpha_3  \;  (\alpha_1\alpha_2 \alpha_3)^{\zeta-1}    \; 
   \frac{e^{ - p^2 \frac{ \alpha_1 \alpha_2 \alpha_3}{ \alpha_1 \alpha_2 + \alpha_1 \alpha_3 + \alpha_2 \alpha_3 } }}{ ( \alpha_1 \alpha_2 + \alpha_1 \alpha_3 + \alpha_2 \alpha_3)^{d/2} }  \crcr
   & =M^{\Lambda}_k(0)- \frac{p^{2d-6\zeta}}{Z^3(4\pi)^d \Gamma(\zeta)^3} f\left( \frac{ k^2 }{p^2}, \frac{p^2}{\Lambda^2} \right)
\end{split}
\ee
with $f$ a function such that $ \lim_{x,y\to 0} f(x,y) = \rm{finite}$.
The plan is:
\begin{itemize}
 \item Taylor expand at order one with integral rest in the variable $tp^2$

 \item Consider the integral rest and rescale $t$ by $p^2$. This yields the 
 scaling in $p$ and $p^2$ appears only in the limits of the integral.
 Introduce Hepp sectors and compute the integral when sending the cutoffs to their limits. 
 \item Compute $M^{\Lambda}_k(0)$.
\end{itemize}
Denoting $I(p^2)=Z^3(4\pi)^d\Gamma(\zeta)^3M_k^{\Lambda}(p)$ in order to get rid of the overall constant we have: 
\begin{align*}
 I(p^2) =&  \int_{\Lambda^{-2}}^{k^{-2}} d\alpha \;\frac{   (\alpha_1\alpha_2 \alpha_3)^{\zeta-1}   }{ ( \alpha_1 \alpha_2 + \alpha_1 \alpha_3 + \alpha_2 \alpha_3)^{d/2} } \crcr
      & -p^2 \int_0^1 dt \int_{\Lambda^{-2}}^{k^{-2}} d\alpha   \frac{   (\alpha_1\alpha_2 \alpha_3)^{\zeta}   }{ ( \alpha_1 \alpha_2 + \alpha_1 \alpha_3 + \alpha_2 \alpha_3)^{1+d/2} }
         e^{-tp^2  \frac{ \alpha_1 \alpha_2 \alpha_3}{ \alpha_1 \alpha_2 + \alpha_1 \alpha_3 + \alpha_2 \alpha_3 } } \;.
\end{align*}
Rescaling $\alpha  = \frac{ \alpha'}{p^2}$ and dropping the primes yields the rest term:
\[
p^{2d-6\zeta} f\left( \frac{ k^2 }{p^2}, \frac{p^2}{\Lambda^2} \right) =  ( p^2)^{1 - 3 -3\zeta +  2 + d} \int_0^1 dt \int_{ \frac{ p^2}{ \Lambda^2} }^{\frac{ p^2}{ k^2}} d\alpha   \frac{   (\alpha_1\alpha_2 \alpha_3)^{\zeta}   }{ ( \alpha_1 \alpha_2 + \alpha_1 \alpha_3 + \alpha_2 \alpha_3)^{1+d/2} }
         e^{-t  \frac{ \alpha_1 \alpha_2 \alpha_3}{ \alpha_1 \alpha_2 + \alpha_1 \alpha_3 + \alpha_2 \alpha_3 } }
\]

We split the $\alpha$ integrals in a sum over 6 Hepp sectors (total orderings of the parameters $\alpha$),  and in the sector $\alpha_1 < \alpha_2 < \alpha_3$, 
we change variables to $\alpha_3 = \rho , \alpha_2 = x  \rho , \alpha_1 =  y x  \rho $. The rest term is then:
\begin{align*}
 f\left( \frac{ k^2 }{p^2}, \frac{p^2}{\Lambda^2} \right)  = &  6 \int_0^1 dt  \int_{ \frac{p^2}{\Lambda^2} }^{\frac{p^2}{k^2}  } d\rho  \;  \rho^{  3 \zeta - d }  
 \int_{\frac{p^2}{\Lambda^2} \rho^{-1} }^1 dx \; x^{2\zeta -\frac{d}{2}   } \int_{ \frac{p^2}{\Lambda^2} \rho^{-1}x^{-1}}^1 dy  \;   \; 
  \frac{  y ^{\zeta }  }{  ( 1 +  y  +  x y )^{1 + d/2} } e^{-t   \frac{ \rho x y  }{ 1 + y + xy }}
\end{align*}
The integrals are clearly convergent when sending $\Lambda \to \infty$ and we get
\begin{align*}
 f\left( \frac{ k^2 }{p^2}, 0  \right)  = &  6 \int_0^1 dt  \int_{ 0 }^{\frac{p^2}{k^2}  } d\rho  \;  \rho^{  3 \zeta - d }  
 \int_{ 0  }^1 dx \; x^{2\zeta -\frac{d}{2}   } \int_{ 0 }^1 dy  \;   
  \frac{  y ^{\zeta }  }{  ( 1 +  y  +  x y )^{1 + d/2} }   e^{-t   \frac{ \rho x y  }{ 1 + y + xy }}
\end{align*}
We will compute this in the next subsection but for now we  check that it is convergent when sending $k\to 0$. As $ 3 \ge 1 + y + xy \ge 1$ we have an upper bound:
\[
  f\left( \frac{ k^2 }{p^2}, 0  \right)  \le   6 \int_0^1 dt  \int_{ 0 }^{\frac{p^2}{k^2}  } d\rho  \;  \rho^{  3 \zeta - d }  
 \int_{ 0  }^1 dx \; x^{2\zeta -\frac{d}{2}   } \int_{ 0 }^1 dy  \;   
   y ^{\zeta }   e^{- \frac{t}{3}   \rho x y  } \;,
\]
and rescaling $ \rho = \frac{ u }{txy }$ and sending $k\to 0$ we get:
\[
 f\le 6 \int_0^1 dt  
 \int_{ 0  }^1 dx \; x^{2\zeta -\frac{d}{2}   } \int_{ 0 }^1 dy  \;   
   y ^{\zeta }  
   \int_{ 0 }^{ \infty  } d u   \;  \left( \frac{1}{txy}\right)^{1 + 3\zeta - d} u^{  3 \zeta - d }  
   e^{- \frac{u}{3}  } \;.
\]
recalling that $\zeta  = d/4$ the bound writes
\[
 f \le 6 \int_0^1 dt \; t^{-1 + \frac{d}{4}} \int_0^1 dx\; x^{-1+ \frac{d}{4}} \int_0^1 dy \; y^{-1 + \frac{d}{2}} \int_0^{\infty} du \; u^{-\frac{d}{4}} e^{-\frac{u}{3}} \;,
\]
which is convergent for $d<4$. For $d=4$, we still need to deal with the last integral.

 \subsection{The integral rest term}
 
 We now compute the numerical constant $f(0,0)$ defined by:
\[
 f\left( 0, 0 \right) =  \int_0^1 dt \int_{ 0 }^{\infty} d\alpha    \;\frac{   (\alpha_1\alpha_2 \alpha_3)^{\zeta}   }{ ( \alpha_1 \alpha_2 + \alpha_1 \alpha_3 + \alpha_2 \alpha_3)^{1+d/2} }
         e^{-t  \frac{ \alpha_1 \alpha_2 \alpha_3}{ \alpha_1 \alpha_2 + \alpha_1 \alpha_3 + \alpha_2 \alpha_3 } } \;.
\]
First, we change variables to $u=\alpha_1\alpha_2,  \; v=\alpha_1\alpha_3, \; 
w=\alpha_2\alpha_3 $ to get:
\begin{equation*}
f\left( 0, 0 \right) =  \frac{1}{2}\int_0^1 dt \int_{ 0 }^{\infty} dudvdw \; \frac{\left(uvw\right)^{\zeta/2-1/2}}{\left(u+v+w \right)^{1+d/2}}e^{-t\frac{\left(uvw\right)^{1/2}}{u+v+w}} \;,
\end{equation*}
and we perform a second change of variables $u=\alpha\delta\gamma, \;
v=\alpha\delta(1-\gamma) , \; w=\alpha(1-\delta)$ to obtain:
\begin{align*}
f\left( 0, 0 \right)=& \frac{1}{2}\int_0^1 dt \int_{ 0 }^{\infty} d\alpha \int_0^1 d\delta \int_0^1 d\gamma \crcr
    & \qquad \alpha^{-1/2+3\zeta/2-d/2}\delta^{\zeta}\gamma^{\zeta/2-1/2}(1-\delta)^{\zeta/2-1/2}(1-\gamma)^{\zeta/2-1/2}e^{-\alpha^{1/2}t\delta(1-\delta)^{1/2}(1-\gamma)^{1/2}} \;.
\end{align*}
The integral over $\alpha$ is a $\Gamma$ function and we use 
$
\int_0^{\infty} dx \; x^a \exp\{-bx^{1/2}\} = 2 b^{ - 2a- 2}\Gamma(2a+2)
$
to write:
\begin{equation*}
f\left( 0, 0 \right)=\Gamma(1+3\zeta-d)\int_0^1 dt t^{d-3\zeta-1} \int_0^1 d\delta \delta^{d-2\zeta-1}(1-\delta)^{d/2-\zeta-1}\int_0^1 d\gamma \gamma^{d/2-\zeta-1}(1-\gamma)^{d/2-\zeta-1} \;.
\end{equation*}
The integral on $t$ gives a factor $ (d-3\zeta)^{-1}$ while the integrals on $\delta$ and $\gamma$ are Beta functions:
\begin{equation*}
f\left( 0, 0 \right)=\frac{1}{d-3\zeta}\frac{\Gamma(1-d+3\zeta)}{\Gamma(3d/2-3\zeta)}\Gamma(d/2-\zeta)^3 \;,
\end{equation*} 
which  for $\zeta=\frac{d}{4}$ simplifies to :
\begin{equation*}
f\left( 0, 0 \right)=\frac{4}{d}\frac{\Gamma(1-d/4)}{\Gamma(3d/4)}\Gamma(d/4)^3 \;.
\end{equation*} 

\subsection{The local part}

We now want to compute the UV divergent piece $M_k^{\Lambda}(0)$.
We denote: 
\begin{align*}
I_0=& Z^3(4\pi)^d\Gamma(\zeta)^3M_k^{\Lambda}(0)=\int_{\Lambda^{-2}}^{k^{-2}} d\alpha \;\frac{   (\alpha_1\alpha_2 \alpha_3)^{\zeta-1}   }{ ( \alpha_1 \alpha_2 + \alpha_1 \alpha_3 + \alpha_2 \alpha_3)^{d/2} } \crcr
 =&\Lambda^{2(d-3\zeta)}\int_{1}^{k^{-2}\Lambda^2} d\alpha \;\frac{   (\alpha_1\alpha_2 \alpha_3)^{\zeta-1}   }{ ( \alpha_1 \alpha_2 + \alpha_1 \alpha_3 + \alpha_2 \alpha_3)^{d/2} } \;.
\end{align*}
This is convergent for $k\rightarrow 0$ so we can write $I_0$ as :
\begin{equation*}
I_0=\Lambda^{2(d-3\zeta)}\int_1^{\infty}d\alpha \;\frac{   (\alpha_1\alpha_2 \alpha_3)^{\zeta-1}   }{ ( \alpha_1 \alpha_2 + \alpha_1 \alpha_3 + \alpha_2 \alpha_3)^{d/2} }-\Lambda^{2(d-3\zeta)}\int_{k^{-2}\Lambda^2}^{\infty}d\alpha \;\frac{   (\alpha_1\alpha_2 \alpha_3)^{\zeta-1}   }{ ( \alpha_1 \alpha_2 + \alpha_1 \alpha_3 + \alpha_2 \alpha_3)^{d/2} }
\end{equation*}

Let us denote :

\begin{equation*}
I_1=\Lambda^{2(d-3\zeta)}\int_1^{\infty}d\alpha \;\frac{   (\alpha_1\alpha_2 \alpha_3)^{\zeta-1}   }{ ( \alpha_1 \alpha_2 + \alpha_1 \alpha_3 + \alpha_2 \alpha_3)^{d/2} }
\end{equation*}

and 

\begin{equation*}
I_2=\Lambda^{2(d-3\zeta)}\int_{k^{-2}\Lambda^2}^{\infty}d\alpha \;\frac{   (\alpha_1\alpha_2 \alpha_3)^{\zeta-1}   }{ ( \alpha_1 \alpha_2 + \alpha_1 \alpha_3 + \alpha_2 \alpha_3)^{d/2} }
\end{equation*}

These integrals can be separated into six Hepp sectors. For the sector $\alpha_1<\alpha_2<\alpha_3$ we can make the change of variables :
\begin{align*}
\alpha_1&=\rho \crcr
\alpha_2&=\rho x \crcr
\alpha_3&=\rho x y 
\end{align*} 

Then, we get for $I_1$, for $\zeta=\frac{d}{4}$ :
\begin{equation*}
I_1= 6\Lambda^{d/2}\int_{1}^{\infty} d\rho \rho^{-d/4-1} \int_1^{\infty}dx \int_1^{\infty}dy \frac{x^{-1}y^{d/4-1}}{\left(1+y+xy\right)^{d/2}} \;,
\end{equation*}
which is
\begin{align*}
I_1=6\Lambda^{d/2}\frac{4}{d}\int_1^{\infty}dx \int_1^{\infty}dy \frac{x^{-1}y^{d/4-1}}{\left(1+y+xy\right)^{d/2}} \;,
\end{align*}

We can do the same change of variables for $I_2$, we get:

\begin{equation*}
I_2= 6\Lambda^{d/2}\int_{k^{-2}\Lambda^2}^{\infty} d\rho \rho^{-d/4-1} \int_1^{\infty}dx \int_1^{\infty}dy \frac{x^{-1}y^{d/4-1}}{\left(1+y+xy\right)^{d/2}} \;,
\end{equation*}

which is :

\begin{align*}
I_2=6\Lambda^{d/2}\frac{4}{d}\left(\frac{\Lambda}{k}\right)^{-d/2}\int_1^{\infty}dx \int_1^{\infty}dy \frac{x^{-1}y^{d/4-1}}{\left(1+y+xy\right)^{d/2}} \;,
\end{align*}

Finally, we obtain :

\begin{align*}
M_k^{\Lambda}(0)=\Lambda^{d/2} 
\frac{24}{dZ^3(4\pi)^d\Gamma(d/4)^3}\left(1-\left(\frac{\Lambda}{k}\right)^{-d/2}\right) \int_1^{\infty}dx \int_1^{\infty}dy \frac{x^{-1}y^{d/4-1}}{\left(1+y+xy\right)^{d/2}}\;.
\end{align*}

For $k \rightarrow 0$ we get:
\begin{align*}
\lim_{k\rightarrow 0}M_k^{\Lambda}(0)=\Lambda^{d/2} 
\frac{24}{dZ^3(4\pi)^d\Gamma(d/4)^3}\int_1^{\infty}dx \int_1^{\infty}dy \frac{x^{-1}y^{d/4-1}}{\left(1+y+xy\right)^{d/2}}\;.
\end{align*}


\providecommand{\href}[2]{#2}\begingroup\raggedright\endgroup


\end{document}